\pdfoutput=1
\documentclass{amsart}
\usepackage[utf8]{inputenc}
\usepackage{amssymb}
\usepackage{amsmath}
\usepackage{graphics}
\usepackage{graphicx}
\usepackage{pifont}	
\usepackage{mathtools}
\usepackage{hyperref}
\usepackage{virginialake}
\vlnostructuressyntax
\usepackage{tikz-cd} 
\usepackage{cmll} 
\usepackage[capitalise]{cleveref}
\usepackage{url}
\usepackage[textsize=footnotesize]{todonotes}
\usepackage{thmtools,thm-restate}
\usepackage{subfloat}
\usepackage{enumerate}
\usepackage{bussproofs}
\usepackage{cmll}
\usepackage{tikz-cd}
\usepackage{algorithm}
\usepackage{algpseudocode}
\usepackage{adjustbox}
\usepackage{xspace}
\usepackage{tikz}
\usetikzlibrary{arrows.meta}
\usetikzlibrary{decorations.markings}
\usetikzlibrary{matrix}
\usetikzlibrary{arrows}
\usetikzlibrary{decorations.pathmorphing}
\usetikzlibrary{decorations.text}
\usetikzlibrary{shapes.geometric}
\usetikzlibrary{shapes}
\usetikzlibrary{arrows}
\usetikzlibrary{positioning} 
\usepackage{tikz-cd} 
\usetikzlibrary{calc}
\usetikzlibrary{shapes}
\usetikzlibrary{arrows}
\usetikzlibrary{matrix}
\usepackage{blkarray}
\usepackage{framed}
\usetikzlibrary{automata}
\usepackage{wrapfig}
\usetikzlibrary{positioning,calc,arrows,shapes,  tikzmark, decorations.pathreplacing}

\title[Reducibility Candidates and Cut Elimination in the Ill-founded Realm]{Making progress: Reducibility Candidates and Cut Elimination in the Ill-founded Realm}

\author{Gianluca Curzi and Graham E.\ Leigh}
\date{}

\newcommand{\dfn}{:=}

\newcommand{\vlhyE}[1]{\global\setbox\vlhybox=\hbox{$#1$}%
	\vlhyaux{\box\vlhybox}}%
    \newcommand{\fv}[1]{\mathrm{FV}(#1)}

\newtheorem{thm}{Theorem}
\newtheorem{lem}[thm]{Lemma}
\newtheorem{prop}[thm]{Proposition}
\newtheorem{cor}[thm]{Corollary}

\newtheorem{fact}[thm]{Fact}

\newtheorem{defn}[thm]{Definition}
\newtheorem{exmp}[thm]{Example}
\newtheorem{rem}[thm]{Remark}

\newcommand{\red}[1]{{\color{red}#1}}
\newcommand{\blue}[1]{{\color{blue}#1}}

\newcommand{\purple}[1]{{\color{black}#1}}
\definecolor{mygreen}{rgb}{0, 0.5, 0}
\newcommand{\mygreen}[1]{{\color{mygreen}#1}}

\definecolor{airforceblue}{rgb}{0.36, 0.54, 0.66}

\newcommand{\cutform}[3]{\mathsf{cut}_{#3}(#1\mid #2)}
\newcommand{\mcut}{\mathsf{mcut}}

\newcommand{\lfp}[1]{\mathsf{lfp}\left(#1\right)}
\newcommand{\gfp}[1]{\mathsf{gfp}\left(#1\right)}
\newcommand{\env}{\mathcal{E}}
\newcommand{\inter}[2]{\llbracket {#1}\rrbracket^{#2}}
\newcommand{\vldr}[3][8]{\vltr{#2}{#3}{\vlhyE{}}{\vlhyE{\hspace{#1pt} }}{\vlhyE{}}}

\newcommand{\Ord}{\mathrm{Ord}}
\newcommand{\subform}{\subseteq}

\newcommand{\FL}{\mathrm{FL}}
\newcommand{\fl}[1]{\FL(#1)}
\newcommand{\flleq}{\preceq_\FL}
\newcommand{\fll}{\prec_\FL}
\newcommand{\fleq}{\approx_\FL}
\newcommand{\flnleq}{\prec_\FL}

\def\cutelim{\rightsquigarrow}
\def\cutelims{\cutelim^*}
\newcommand{\depth}[2]{\mathsf{depth}_{#1}{(#2)}}

\newcommand{\unit}{\mathbf{1}}
\newcommand{\infrule}{\mathsf{r}}
\newcommand{\rrule}{\infrule} %legacy
\newcommand{\id}{\mathsf{id}}

\newcommand{\ex}{\mathsf{e}}
\newcommand{\cut}{\mathsf{cut}}

\newcommand{\choice}[2]{#2^{#1}}

\newcommand{\zero}{\mathbf{0}}

\newcommand{\muLJ}{\mu \mathsf{LJ}}
\newcommand{\muMALL}{\mu \mathsf{MALL}}

\newcommand{\muLL}{\mu \mathsf{LL}}

\newcommand{\seq}[2]{\mathcal{S}_{#1}({#2})}
\newcommand{\rul}[2]{\mathcal{R}_{#1}({#2})}

\newcommand{\treem}[2]{\mathrm{T}_{#2}(#1)}

\newcommand{\initm}[2]{\mathrm{I}_{#2}(#1)}

\newcommand{\PDer}[1]{\mathfrak{P}_{#1}}
\newcommand{\GDer}[1]{\mathfrak{E}_{#1}}
\newcommand{\ortcoh}[3]{ {#2}\scoh_{#1}{#3}}

\newcommand{\reds}[1]{\mathfrak{#1}}

\newcommand{\finseq}[1]{{#1}^{<\omega}}
\newcommand{\length}[1]{|#1|}
\newcommand{\branch}[1]{\mathfrak{#1}}
\newcommand{\Branch}[1]{\mathrm{Br}(#1)}
\newcommand{\Cnor}[1]{\mathfrak{N}_{#1}}
\newcommand{\Rset}{\mathfrak{R}}

\begin{document}
\maketitle

\begin{abstract}
Ill-founded (or non-wellfounded) proof systems have emerged as a natural framework for inductive and coinductive reasoning. In such systems, soundness relies on global correctness criteria, such as the \textit{progressivity}  condition. Ensuring that these criteria are preserved under infinitary cut elimination remains a central technical challenge in ill-founded proof theory.

In this paper, we present two cut elimination arguments for ill-founded $\muMALL$ -- a fragment of linear logic extended with fixed-points -- based on the reducibility candidates technique of Tait and Girard. In both arguments,  preservation of progressivity  follows directly from the defining properties of the reducibility candidates. In particular, the second argument is derived from the topological notion of internally closed set developed in previous work by Afshari and Leigh.
\end{abstract}

\section{Introduction}

Since their early development, infinitary proofs have played a prominent role in modern proof theory. Their systematic study was initiated by Sch\"{u}tte, who showed cut elimination  for an  infinitary extension of Peano arithmetic  \cite{Schutte1,Schutte2}, where the presence of the $\omega$-rule -- an infinitely branching rule replacing the induction scheme -- makes the {corresponding} notion of proof infinite in \emph{breadth} while remaining well-founded. {Specifically, Sch\"{u}tte’s cut elimination argument was formulated as a transfinite process indexed by ordinals, with termination ensured by a well-founded measure on proof height. This framework was later refined by Mints who emphasised the role of continuity in proof transformations~\cite{mints1978finite}.}

Since this pioneering work, infinitary proof systems and their cut elimination techniques have been extensively studied, providing a valuable setting for establishing consistency proofs and conservativity results for theories of arithmetic and analysis (see, e.g., \cite{Buchholz:1981IteratedInductive,Takeuti:1967ConsistencyProofs,Tait:1970ApplicationsCut}).

Over the past two decades, however, alternative notions of infinitary proof have emerged that relax the requirement of well-foundedness (while often reintroducing finite branching).  Motivated by the study of first-order inductive definitions~\cite{Broth07,brotherston2011sequent}, but anticipated in the seminal work of Niwinski and Walukiewicz on the modal $\mu$-calculus \cite{niwinski1996games}, ill-founded (or non-wellfounded) proofs offer an ideal proof-theoretic framework for  inductive and coinductive reasoning. Research on ill-founded proofs has been applied to a wide range of areas, including  automata theory \cite{Automata}, games~\cite{SANTOCANALE2002305}, theories of arithmetic \cite{Simpson17}, (modal) $\mu$-calculi \cite{sprengerdam03:journal},  complexity~\cite{CurziCIC},  and type theory \cite{Das2021}.

In  ill-founded proof systems,  soundness is no longer guaranteed by local inductive arguments. Instead, it appeals to   \emph{global correctness criteria}. Among these criteria, the most commonly adopted in the literature is the \emph{progressivity} (or \emph{trace}) condition,  typically satisfied whenever  a syntactically determined feature occurs infinitely often {along each  infinite branch of a proof}.
{Other examples of criteria include  bouncing threads, induction orders, semantic productivity, and automata-based conditions  (see, e.g., \cite{sprengerdam03:journal,sprengerdam03:conf, bouncing}). 

A major focus of ill-founded proof theory is on \emph{circular} or \emph{cyclic} proofs, that is, proofs whose underlying tree is (possibly infinite but) \emph{regular} -- i.e.,  having  only finitely many distinct subtrees. Owing to their finite presentability, cyclic proofs are typically expressed as finite, possibly cyclic, graphs and can be viewed as a natural ill-founded counterpart of  traditional inductively presented proofs.}

A growing body of research in this topic is devoted to the family of \emph{fixed-point logics}, which integrate  finitary logics with  least and greatest fixed points. Notable examples include the family of  modal $\mu$-calculi~\cite{Kozen}, fixed-point formulations of intuitionistic logic ($\muLJ$)~\cite{Clair,muLJ} and of  linear logic and its fragments (such as $\muLL$ and $\muMALL$) \cite{BaeldeM07,Baelde2016InfinitaryPT}. When nesting and interleaving of fixed points are permitted, fixed-point logics are well suited to express very general forms of (co)inductive definitions and corresponding (co)inductive reasoning principles. From a computational perspective, fixed-point logics that admit a constructive interpretation -- such as $\muLJ$, $\muLL$ and $\muMALL$ -- can be used to model sophisticated (co)inductive data structures and (co)recursion mechanisms. A detailed account of the computational strength of these logics (both in the inductive and cyclic presentation)  can be found in \cite{muLJ}.

This paper  focuses on $\muMALL$, the fixed-point extension of \emph{multiplicative additive linear logic} ($\mathsf{MALL}$), which can be  obtained from classical logic by removing the structural rules of contraction and weakening. Its key feature is the presence of two (non-equivalent) formulations of the connectives $\wedge$ and $\vee$: the  \emph{multiplicatives}  (resp., $\otimes$ and $\parr$), and the  \emph{additives}  (resp., $\with$ and $\oplus$). The ``exponentials'' of linear logic, which control the use of weakening and contraction, can be simulated by fixed points in \(\muMALL\)~\cite{BaeldeM07}, thus making this logic a general framework for studying  the theory of (co-)induction. 

The introduction of ill-founded proofs challenges one of the central tools of proof theory: cut elimination. When proofs can be infinite in \emph{depth},  Sch\"{u}tte-style  termination arguments are no longer applicable, and a conceptual shift becomes necessary. One approach is to reformulate cut elimination as an infinitary -- possibly transfinite, ordinal-indexed -- rewriting process that converges to a cut-free proof in the limit. This perspective draws on notions and techniques from topology (such as continuity and metric completion) or  coinductive reasoning.  In particular, when the rewriting process has a countable number of steps,   the existence of the limit proof boils down to showing that all finite approximations can be computed in finitely many steps, a property commonly referred to as \emph{productivity}.  

{The use of infinite rewriting for cut elimination in an ill-founded setting is not, in itself, a novel idea. Infinitary rewriting techniques have been studied systematically since the 1990s in variety of contexts, including  first-order and  higher-order rewrite systems~\cite{DERSHOWITZ199171,KETEMA2011893}, and infinitary extensions of the $\lambda$-calculus \cite{Lambda}. What ill-founded proof theory brings to this topic, rather, lies in the  proof-theoretic methods developed to ensure that the infinitary cut elimination procedures preserve the global correctness conditions mentioned above, thereby yielding sound cut-free limit proofs whenever they are applied to sound proofs.

One of the earliest and more general approaches to cut elimination for cyclic and ill-founded proofs was introduced by Santocanale in~\cite{FortierS13} for a fragment of $\muMALL$, and was later extended to full $\muMALL$ by Baelde et al in~\cite{Baelde2016InfinitaryPT}. The core idea is to define a (possibly infinite) rewriting strategy that stepwise pushes cuts upward, giving priority to the bottommost ones. To avoid technical complications such as cut permutations, this approach often resorts to the \emph{multicut} rule, a macro rule representing a series of consecutive cuts.  

{In recent years, a variety of alternative cut elimination techniques have been proposed, such as  the use of fixed-point theorems for contractive maps over ultrametric spaces~\cite{modal1,modal2,modal3}, notions of run~\cite{das:pous:non-well}, or domain-theoretic approaches~\cite{acclavio_et_al:LIPIcs.CSL.2024.8}. Nonetheless, many of these methods rely on features specific of the systems considered, and do not readily extend to different or more general settings.}

Proving that infinitary cut elimination preserves the aforementioned global correctness criteria remains a technically challenging -- yet central -- problem in ill-founded proof theory. Existing arguments in the literature tend to be bespoke and system-specific, and therefore lack full generality and robustness    {\cite{modal1,modal2,modal3,Saurin,das:pous:non-well,acclavio_et_al:LIPIcs.CSL.2024.8}}.

The present article contributes to the general theory by adapting Tait and Girard’s celebrated reducibility candidates technique to the ill-founded setting. More specifically, we isolate two
notions of reducibility candidate. The first, named \emph{$\Cnor{}$-reducibility candidates}, is defined directly in terms of cut elimination. The second, \emph{$\GDer{}$-reducibility candidates}, is based on an alternative, topologically motivated and logic-independent, global condition on proofs which we call \emph{external progressivity}. In essence, the latter are the ill-founded proofs which explicitly present an invariant for cut elimination. Our work is developed within ill-founded $\muMALL$ whose linear structure offers a streamlined and concise exposition, but covers classical and intuitionistic fixed point logics by virtue of the standard embeddings.

We prove that progressing proofs belong to both kind of reducibility candidate, thereby confirming that progressing proofs are both externally progressing and normalisable. The argument relies on a modest generalisation of the standard soundness (or computational totality) argument from ill-founded proof theory (see, e.g.,~\cite{muLJ,farzad-infinitary}).  In the case of $\Cnor{}$-reducibility candidates, the result directly implies cut elimination for progressing proofs. An (almost) parallel argument establishes that every progressing proof is externally progressing, generalising an observation from~\cite{Graham-ICset}. As with the standard reducibility candidates arguments for inductive systems, $\Cnor{}$-reducibility merely confirms the \emph{existence} of a cut elimination strategy. A more insightful and explicit cut elimination argument is developed using $\GDer{}$-reducibility candidates, where external progressivity provides a direct and straightforward certification of preservation of progressivity when a concrete cut elimination procedures is given.

\subsection*{Related work}
Cut elimination for infinitary $\muMALL$ was first developed by Baelde et al.~\cite{Baelde2016InfinitaryPT} using the notion of multicut and extending a previous result for a much weaker system by Fortier and Santocanale~\cite{FortierS13}. To ensure productivity and preservation of progressivity, however, the authors   propose  a bespoke argument based on a truth semantics. 

The concept of externally progressing ill-founded proofs appeared in~\cite{Graham-ICset} as an intermediate notion between progressing and normalisable ill-founded proofs in the context of higher-order intuitionistic arithmetic with fixed-points. Both inclusions -- progressing proofs are externally progressing and externally progressing proofs are normalisable -- crucially relied on restricting the use of fixed point quantifiers.

In contrast to~\cite{Graham-ICset}, our result does not place any specific restriction on the formation of fixed points (beyond positivity). Moreover, thanks to the robustness of the reducibility candidates method and the logic-independent nature of external progressivity, our cut elimination arguments generalise to fixed point logics beyond those treated in~\cite{Baelde2016InfinitaryPT,FortierS13}, such as the higher-order fixed point logics discussed in~\cite{Graham-ICset}.

\subsection*{Outline of the paper} 
\Cref{sec:muMALL} recalls the ill-founded proof system for $\muMALL$, its progressivity condition, and infinitary cut elimination. \Cref{sec:good} introduces  the notion of internally closed set, which is used to define the \emph{externally progressing} derivations. 
\Cref{sec:totality-candidates} is devoted to reducibility candidates and their properties. Finally, in  \Cref{sec:cut-elimination-theorem} and \cref{sec:second-cut-elimination} we prove cut elimination (\Cref{thm:cut-elimination}) by two reducibility-style arguments, one based on $\omega$-normalisation (\Cref{lem:progressing-implies-reducible}) and the other through external progressivity (\Cref{lem:good-norm}).

\subsection*{Conventions on sequences and trees}

 The set of natural numbers is denoted $\omega$, and is associated with the standard ordering $<$. 
A finite sequence over a set $X$ is a function $u:\{0, \ldots, n-1\}\to X$, often written $\langle u(0), \ldots, u(n-1)\rangle$. The cardinality of the domain of $u$ is referred to as the \emph{length} and denoted $\length{u}$. The set of finite sequences over $X$ is denoted $\finseq{X}$. 
The empty sequence is denoted  $\langle\rangle \in \finseq{X}$, and  concatenation of (finite) sequences  $u$ and $v$ is denoted $u. v$. \purple{The prefix relation on $\finseq{X}$ is denoted $\leq$, defined by $u \leq v$ iff there exists $w \in \finseq{X}$ such that $v=u.w$.}

A \emph{tree} is a non-empty  set $T \subseteq \finseq{\omega}$ which is  prefix closed, that is, if $u <v\in T$ then $u \in T$. Elements of $T$ are called \emph{nodes}  of $T$. 
A \emph{leaf} is a \purple{$<$-maximal} node.
A \emph{tree over  $X$} is a pair $T'=(T, \lambda)$ where $T$ is a tree and $\lambda \colon T \to X$ is a label function assigning to each node of $T$ an element of $X$. When there is no cause for confusion we write $u \in T'$ in place of $u \in T$, and $T'(u)$ for $\lambda (u)$. 
Given $v \in T$, the \emph{\( v \)-rooted subtree} of \( T \) is the subtree \( T_v \) with nodes $\{u \mid v.u \in T\}$ and $T_{v}(u)=T(v.u)$.

\purple{In what follows, we will assume the so-called \emph{tree topology} of (labelled) trees. It is well-known that the latter is the topology of the complete (ultra)metric space over the set of (labelled) trees given by the distance function $\delta$ such that 
$\delta(T, T') = 0$ if $T = T'$, and otherwise $\delta(T, T')$ is the infimum of all $2^{-n}$ such that $T$ and $T'$ coincide on all nodes with length $\le n$. This will ensure that Cauchy sequences of derivations have \emph{limits}.}

A \emph{branch} of $T$ is an infinite sequence of successors in $T$ starting from the root, namely a sequence $\mathfrak b: \omega \to \omega^{<\omega}$ such that $\branch b (0)= \langle \rangle$ and $\branch b(n+1)$ is in $\{u \in T \mid \branch b(n)\leq u \text{ and }|u|=n+1\}$.
The set of branches of $T$ is denoted $\Branch T$. 
The longest common prefix of distinct branches $\branch a , \branch b$  of $  T$ is the node $\branch a \wedge \branch b\in T$
. 

Let $B \subseteq \Branch T$ be a non-empty set of branches of a tree $T$. The \emph{closure} of \( B \) is the set \( \overline B = \{ \branch b \in \Branch T \mid \forall n > 0 \exists \branch a \in B \, \branch{a}(i) = \branch b(i) \text{ for all } i < n \} \).
If \( B = \overline B \) we call \( B \) \emph{closed}. Note that \( \Branch T \) is closed if \( T \) is finitely branching.
The \emph{infimum} (\emph{supremum}) of $B$ is the leftmost (resp.~rightmost) branch of $\overline B$, denoted $\inf B$ ($\sup B$). Infima always exist, though suprema are only guaranteed in the case of finitely branching trees.

\section{Ill-founded $\muMALL$}\label{sec:muMALL}

We begin by introducing an ill-founded (one-sided) sequent calculus system for $\muMALL$, the multiplicative-additive fragment of propositional linear logic extended with least and greatest fixed point operators~\cite{BaeldeM07,Baelde2016InfinitaryPT}. Our presentation will mainly follow~\cite{Baelde2016InfinitaryPT}, although we  adopt some terminology from~\cite{Graham-ICset}.

\subsection{Formulas}\label{subsec:formulas}

\begin{defn}[Formulas]
\emph{Preformulas}, written $\phi, \psi$ etc., are generated by the following grammar:
\[
    \phi, \psi \dfn X \mid \bot \mid \unit \mid \phi \parr \psi \mid \phi \otimes \psi   \mid \zero \mid \top \mid \phi \oplus \psi \mid \phi \with \psi \mid \mu X \phi \mid \nu X \psi
\]
where $X$ belongs to a countable set of (\emph{propositional}) \emph{variables}.

\emph{Free variables} of a preformula are defined as expected, construing $\mu$ and $\nu$ as binders:
\begin{itemize}
    \item $\fv X \dfn \{X\}$
    \item $\fv \bullet \dfn \emptyset$, for $\bullet \in \{\zero, \top, \bot, \unit\}$
    \item $\fv {\phi \star \psi} \dfn \fv \phi \cup \fv \psi$, for $\star \in \{\oplus,\otimes,\parr, \with\}$
    \item $\fv{\kappa X \, \phi} \dfn \fv \phi \setminus \{X\}$, for $\kappa \in \{\mu,\nu\}$
\end{itemize}
A preformula is \emph{closed} if it has no free variables, otherwise it is \emph{open}. Closed preformulas are called \emph{formulas}.
\end{defn}

Capture-avoiding substitution of a formula $\psi$ for a free variable $X$ in $\phi$, written $\phi[\psi/X]$, is defined in the standard way. We will sometimes write $\psi(\mu X\phi)$ for $\psi[\mu X \phi/X]$, and similarly for $\nu$.
We also assume some standard conventions on variable binding, in particular that each occurrence of a binder $\mu$ or $\nu$ binds a variable distinct from all other binder occurrences in consideration. This avoids having to deal with variable renaming explicitly.

 \emph{Negation} of a formula $\phi$, written $\phi^\perp$,  is defined as  the involution on preformulas satisfying $(\phi \parr \psi)^\perp= \phi^\perp \otimes \psi^\perp$, $(\phi \oplus \psi)^\perp=  \phi^\perp \with \psi^\perp$, $\perp^\perp=
\unit$, $\zero^\perp= \top$, $(\nu X \phi)^\perp=  \mu X \phi^\perp$, $X^\perp= X$. Notice that the equations imply  $(\phi[\psi/X])^\perp= \phi^\perp [\psi^\perp/X]$.\footnote{The equation $X^\perp= X$ is harmless since throughout this paper we will only consider closed preformulas. Thanks to that equation, we  do not require any positivity condition on fixed point formulas to enforce a semantic interpretation based on the Knaster-Tarski fixed point theorem.}

The presence of fixed point formulas requires a more permissive notion of subformula, given by the well-known Fischer-Ladner preorder.

\begin{defn}[Fischer-Ladner preorder] We  write $\phi \subform \psi$ if $\phi$ is a subformula of $\psi$. The \emph{Fischer-Ladner preorder}, written $\flleq$,
is the smallest reflexive and transitive extension of $\subform$ satisfying
$\phi(\mu X \phi(X)) \flleq \mu X \phi(X)$ and $\phi(\nu X \phi(X)) \flleq \nu X \phi(X)$.
We write $\phi \fleq \psi$ if $\phi \flleq \psi \flleq \phi$, and define ${\fll} = {\flleq} \setminus {\fleq}$.
The \emph{Fischer-Ladner closure} of a formula $\psi$, written $\fl \psi$, is the set $\{\phi \mid \phi \flleq \psi\}$.
\end{defn}

Note that  $\fleq$-equivalence classes are naturally (well) partially ordered by $\flleq$. Notice also that  $\fl \psi$ is the smallest set of formulas closed under subformulas and fixed point unfolding: whenever $\kappa X\phi(X) \in \fl \psi$, then also $\phi(\kappa X\phi(X)) \in \fl \psi$ for $\kappa \in \{\nu, \mu\}$. Also,  $\fl \psi$ is a \emph{finite} set.    

The Fisher-Ladner preorder allows   a standard 
  (strict) well partial order on formulas, $<$, that associates a ``priority'' to interleaving fixed points within a formula (see, e.g.,~\cite{Studer08,Doumane17thesis,muLJ}).

\begin{defn}
	[Priority] \label{defn:priority}
	We say that $\phi$ has higher \emph{priority} than $\psi$, written $\psi<\phi$, if $\psi \flnleq \phi$, or $\phi \fleq \psi $ and $\phi \subset \psi$.
\end{defn}

\begin{exmp}
    Consider the following formula $\phi= \mu X\psi$, where $\psi=\nu Y(X \otimes Y \otimes \theta)$ and $\theta= \nu Z (\unit \oplus Z)$. Setting $\chi= \nu Y (\phi \otimes Y \otimes \theta)$ we have that $\chi \flleq \phi$. Moreover, since $\phi \subset \chi$, we also have $\phi \flleq \chi$, and so $\chi \fleq \phi$. On the other hand, since $\theta \subset \phi$ we  have $\theta \flleq \phi$ (but not  $\phi  \flleq \theta$). Finally, notice that the preformula $\psi$ is not closed, and so it is not a formula. This means that $\psi$ is incomparable with respect to any of these orderings. 
\end{exmp}

\begin{lem}
	If $\phi \fleq \psi $ then $\psi<\phi$ iff \( \phi \) is shorter than \( \psi \).
\end{lem}

In what follows, we will assume an arbitrary extension of $<$ to a \emph{total} well order.

\subsection{Derivations}

We can now define an infinitary proof system for $\muMALL$. First, we present the one-sided sequent calculus rules. Derivations in this system are  possibly infinite labelled trees constructed from those rules. 

\begin{defn}[Sequents and inference rules]
    A \emph{sequent} of $\muMALL$ is a finite list of formulas. The inference rules of the sequent calculus for ill-founded $\muMALL$ are presented in \Cref{fig:sequent-calculus-MALL}, and are split into two categories: the \emph{logical rules} ($\top$, $\unit$, $\bot$, $\otimes$, $\parr$, $\choice 0 \oplus$, $\choice 1\oplus$, $\with$, $\nu$, $
\mu$),  and the \emph{structural rules} ($\cut$, $\ex$). In particular, the rules $\top$, $\zero$, $\choice 0 \oplus, \choice 1\oplus, \with$ are called \emph{additive}, the rules $\unit, \bot, \otimes, \parr$ are called \emph{multiplicative}, and the rules $\mu, \nu$ are called \emph{fixed point}.
The sequent displayed below the inference line is the \emph{conclusion} and the sequents above are the \emph{premise(s)}. In $\cut$ the formula $\phi$ is the \emph{cut formula}. $\Gamma$ and $\Delta$ designate lists of formulas called \emph{contexts}. In the logical rules a single formula in the conclusion (i.e., the one not in the context) is designated as the \emph{principal} formula, and the distinguished formulas in the premise(s) as \emph{minor} formulas. 
\end{defn}

\begin{figure}[t]
    \centering
    \(
    \begin{array}{c}
   \vlinf{\ex}{}{{\Gamma}, \psi, \phi, {\Delta}}{{\Gamma},{\phi}, \psi ,{\Delta}}
   \qquad
    \vliinf{\cut}{}{{\Gamma}, \Delta}{\Gamma,\phi}{\phi^\perp, \Delta}
   \\[1.25em]
   \vlinf{\unit}{}{ \unit}{}
   \qquad 
   \vlinf{\perp}{}{{\Gamma},\perp}{{\Gamma}}
   \qquad 
   \vliinf{\otimes}{}{ {\Gamma} , {\Delta},{\phi \otimes \psi}}{ {\Gamma},\phi }{{\Delta},\psi }
   \qquad 
   \vlinf{\parr}{}{ {\Gamma},{\phi \parr \psi} }{{\Gamma},{\phi},{\psi} }
   \\[1.25em]
   \vlinf{\top}{}{ \Gamma,\top}{}
   \qquad
   \vlinf{\choice 0 { \oplus}}{}{{\Gamma},{\psi_0 \oplus \psi_1} }{ {\Gamma},{\psi_0} }\qquad
    \vlinf{\choice 1 { \oplus}}{}{{\Gamma},{\psi_0 \oplus \psi_1} }{{\Gamma},{\psi_1} 
    }
    \qquad 
    \vliinf{\with}{}{{\Gamma},{\phi \with \psi} }{{\Gamma},{\phi}  }{{\Gamma},{\psi} }
     \\[1.25em]
       \vlinf{ \mu}{}{{\Gamma},{\mu X \phi(X)}}{{\Gamma},{\phi(\mu X \phi(X))}}
       \qquad 
        \vlinf{ \nu}{}{{\Gamma},{\nu X \phi(X)}}{{\Gamma},{\phi(\nu X \phi(X))} }
    \end{array}
    \)
    \caption{Sequent calculus rules for $\muMALL$ (infinitary presentation).}
    \label{fig:sequent-calculus-MALL}
\end{figure}

\begin{defn}[Derivations] A \emph{derivation} is a (possibly infinite) tree $d=(T, \lambda)$ over pairs of sequents and rule names consistent with the inference rules in \Cref{fig:sequent-calculus-MALL}. The sequent and rule names at a node $u \in d$  are denoted  $\seq d u$ and $\rul d u$ respectively, so $\lambda (u)= (\seq d u , \rul d u)$. Consistency with the inference rules means that for every $u \in d$, the sequent $\seq d u$ occurs as the conclusion of an instance of the rule $\rul d u$ whose premises are the sequents associated to successors of $u$ (in order). The sequent labelling the root is called \emph{conclusion}. A \emph{subderivation} of $d$ is a subtree $d_u$ of $d$ rooted at a node $u \in d$.
\end{defn}

\begin{exmp}\label{exmp:derivations}
    \Cref{fig:example-derivations} illustrates two examples of derivations (colours
may be ignored for now). Notice that $\Gamma$ can be any formula, possibly $\perp$. Therefore derivations are not logically sound.
\end{exmp}
\begin{figure}
    \centering
   \(
   \vlderivation{
\vliin{\cut}{}{\Gamma}
{\vlin{\nu}{}{\Gamma, \red{\nu X.X}}{\vlin{\nu}{}{ \Gamma, \red{\nu X.X}}{\vlin{\nu}{}{ \Gamma, \red{\nu X.X}}{\vlhy{\branch{b}\, \vdots}}}}}
{\vlin{\mu}{}{\blue{\mu X.X}}{\vlin{\mu}{}{\blue{\mu X.X}}{\vlin{\mu}{}{\blue{\mu X.X} }{\vlhy{ \vdots\, \branch{b}'}}}}}
}
\qquad \qquad
\vlderivation{
\vlin{\sigma}{}{\mygreen{A}}{
\vlin{\nu}{}{\mygreen{B(A)}}{
\vlin{\choice 0 \oplus}{}{\mygreen{A \oplus B(A)}}
{
\vlin{\sigma}{}{\mygreen{A}}{\vlhy{\vdots}}
}
}
}
}
\)
    \caption{Examples of derivations, where $A\dfn \sigma X. B(X)$ and $B(X)\dfn \nu Y (X \oplus Y)$, with $\sigma \in \{\nu, \mu\}$.}
    \label{fig:example-derivations}
\end{figure}

To simply the presentation of many concepts, henceforth we largely leave occurrences of the exchange rule \( \ex \) implicit. This convention is especially relevant in the case of cut which will often be assumed to be of the following ``exchange invariant'' form:
\[
   \vliinf{\cut}{}{{\Gamma}, \Gamma', \Delta, \Delta'}{{\Gamma},\phi, \Gamma'}{\Delta,\phi^\perp, {\Delta'}}
\]
Ignoring applications of exchange  will allow us to denote a derivation obtained by applying a consecutive series of cut rules as follows:
\[
\vlderivation{
\vliin{\cut}{}{\Gamma, \Delta_1, \ldots, \Delta_n}
{
\vliin{}{}{\vdots}{
  \vliin{\cut}{}{\Gamma, \Delta_1, \phi_2, \ldots, \phi_n}
{ 
\vltr{d}{\Gamma, \phi_1, \ldots, \phi_n}{\vlhy{}}{\vlhy{}}{\vlhy{}}
}{
\vltr{d_1}{\phi_1^\perp,\Delta_1}{\vlhy{}}{\vlhy{}}{\vlhy{}}
}
}
{\vltr{d_2}{\phi_2^\perp,\Delta_2}{\vlhy{}}{\vlhy{}}{\vlhy{}}}
}{
\vltr{d_n}{\phi_n^\perp,\Delta_n}{\vlhy{}}{\vlhy{}}{\vlhy{}}
}
}
\]
We will denote the above derivation with  $\cutform{d}{d_1, \ldots, d_n}{\phi_1, \ldots, \phi_n}$, or simply $\cutform{d}{d_1, \ldots, d_n}{}$ when the cut formulas are clear from the context.

\subsection{Threads, traces, and progressivity}

\Cref{exmp:derivations} shows, among others,  that the system of derivations of $\muMALL$ is not logically sound. A typical approach to recover soundness is to introduce a global correctness criterion called \emph{progressivity} (or \emph{trace}) \emph{condition} (see, e.g.,~\cite{Broth07}). Intuitively, the progressivity condition certifies that along each branch of a given derivation we can find a sequence of formulas ordered by $\fll$ where certain $\nu$-formulas unfold infinitely often. Such sequences are called \emph{threads}.

\begin{defn}[Threads]
A \emph{weak thread} is an infinite sequence $\tau$ such that $\tau(i)$ is a pair $(\phi_i, p_i)$, where $\phi_i$ is a formula and $p_i\in\{0,1\}$, for which for all $i$ either $(\phi_i, p_i) = (\phi_{i+1}, p_{i+1})$ or $\phi_{i+1} \fll \phi_i$ and:
\begin{itemize}
    \item if $\phi_i=\kappa X\psi$ then $p_{i+1}=0$
    \item if $\phi_{i}= \psi_0 \star \psi_0$ with $\star \in \{\oplus,\otimes,\parr, \with\}$ then $\phi_{i+1}=\psi_{p_{i+1}}$.
\end{itemize}
Henceforth, for the sake of readability, we will consider $\tau$ as a sequence of formulas, thus treating each $\tau(i)$ as a formula rather than a pair.

 A \emph{thread} is a weak thread $\tau$ such that $ \tau(i+1)\fll \tau(i) $ for all $i$. A weak thread that is not eventually constant uniquely identifies a thread $\hat \tau$ given by contracting all consecutive repetitions, in which case we refer to $\tau$ as the \emph{expansion} of $\hat \tau$. Finally, the  \emph{dual} of a (weak) thread $\tau=(\phi_i)_i$, written $\tau^\perp$, is the (weak) thread $(\phi_i^\perp)_i$.
\end{defn}

The following is a well-known property of threads (see, e.g.,~\cite{CmuPA}):

\begin{fact}\label{lem:good-formula}
    Let $\tau=(\phi_i)_i$ be a thread. There is a unique $\kappa X \psi$ with $\kappa \in \{\mu, \nu\}$ such that:
    \begin{itemize}
        \item for infinitely many $i$, $\phi_i= \kappa X \psi$, 
        \item for all but finitely many $i$, $\phi_i$ contains $\kappa X\psi$ as a subformula.
    \end{itemize}
\end{fact}

The above property allows us to identify those threads  that can be used to certify progressivity condition, which we call ``good''.

\begin{defn}[Good threads]\label{defn:good-and-bad}
    Let $\tau=(\phi_i)_i$ be a thread, and let $\kappa X \psi$ be a formula given by \cref{lem:good-formula}. If $\kappa= \nu$ we call $\tau$ \emph{good}; otherwise $\tau$ is \emph{bad}.
\end{defn}

A consequence of the above is that a thread $(\phi_i)_{i \geq 0}$ is good (bad) iff $(\phi_i)_{i \geq n}$ is good (bad) for every \( n \ge 0 \). Moreover:

\begin{prop}\label{lem:good-iff-bad}
    A thread is either good or bad, and it is good iff its dual is bad.
\end{prop}

We now need  the notion of \emph{trace}. This  can be seen as a concrete instantiation of a thread $\tau$ along a branch, and is defined as a sequence of  \emph{formula occurrences}  enumerating (with possible repetitions) the formulas of $\tau$.

\begin{defn}[Formula occurrences] Let $d$ be a derivation. A \emph{formula occurrence (of $d$)} is a triple $(u, k, \phi)$ where $u \in d$ and  $k< \length{\seq d u}$ is such that $\seq d u (k)=\psi_1, \ldots, \psi_n$ with $\psi_k=\phi$. 
\end{defn}

\begin{defn}[Ancestry] Let $d$ be a derivation.
{We say that a formula occurrence $f= (u, k, \phi)$ is an \emph{ancestor} of another formula occurrence $f'= (v, h, \psi)$ if one of the following cases holds:
\begin{itemize}
    \item $f'$ is principal and $f$ is a minor formula
    \item  $f'$  is in the context, and $f$ is the corresponding formula in a premise of the rule. 
\end{itemize}
}
\end{defn}

\begin{defn}[Trace] A \emph{trace} in $d$ is an infinite sequence $t=(u_i, k_i, \phi_i)_i$ of formula occurrences  such that {$t(i+1)$ is an ancestor of $t(i)$} for every $i<\omega$.
The unique branch $\branch b \in \Branch{d}$ which contains $(u_i)_i$ is said to \emph{carry} $t$.
\end{defn}

Notice that, if $t=(u_i, k_i, \phi_i)_i$  is a trace in $d$ then $ \phi_{i+1}\flnleq\phi_i $ if $t(i)$ is principal in a logical rule and $\phi_i=\phi_{i+1}$ otherwise.

\begin{defn}
    Let $\tau$ be a weak thread. A \emph{trace of $\tau$} is a trace $t=(u_i, k_i, \phi_i)_i$ such that $(\phi_i)_i$ is an expansion of $\tau$. A branch  $\branch b\in \Branch d$ \emph{bears} $\tau$, equivalently $\tau$ is \emph{borne} by $\branch b$, if $\tau$ has a trace carried by $\branch b$. We say that a set of branches \emph{bears} a thread if one of its branches does so.  
\end{defn}

By weak K\"{o}nig's lemma, every branch of a derivation carries at least one trace. As a consequence, every branch bears at least one weak thread though not necessarily a (non-weak) thread. Indeed, we can state:

\begin{lem}\label{lem:non-steble-traces-implies-bearing-thread}
    A branch bears a thread iff the branch carries a trace which is infinitely often principal in a logical rule.
\end{lem}

Thanks to the notions of thread and trace, we can now define the global correctness condition for derivations of ill-founded $\muMALL$:

\begin{defn}[Progressing derivations] A \emph{progressing} derivation is a 
derivation for which every branch bears a good thread. 
 With  $\PDer{}$  we denote the set of progressing derivations.
\end{defn}

\begin{exmp}[\Cref{fig:example-derivations}, revisited]
 Consider the leftmost derivation of~\Cref{fig:example-derivations}. It has only two branches $\branch b$ and $\branch{b'}$ defined by $\branch{b}(n)= \langle 0, \overset{n}{\ldots}, 0 \rangle$ and $\branch{b}(n)= \langle 1, 0, \overset{n-1}{\ldots}, 0 \rangle$, respectively.
 \purple{Let $\tau$ (resp., $\tau^\perp$) be the  thread bourne by $\branch{b}$ (resp., $\branch{b}'$) and whose trace is highlighted with \red{red}  (resp., \blue{blue}) formula occurrences.}  Since $\tau^\perp$ is the only thread borne by  $\branch{b}'$ and it is not good, then the derivation is not progressing.

\purple{Concerning the centre derivation, its only branch $\branch b$   carries the trace highlighted with \mygreen{green} formulas. Its thread $\tau$ unfolds two fixed points, $A$ and $B(A)$, the smallest one is $A$. Then, $\tau$ is good precisely when  $\sigma=\nu$, in which case the derivation is progressing. }
\end{exmp}

\begin{figure}
    \centering
\(
\id_{\unit}\dfn 
\vlderivation{
\vlin{\perp}{}{\perp, \unit}{\vlin{\unit}{}{\unit}{\vlhy{}}}
}
\qquad
\id_{\top}\dfn 
\vlderivation{
\vlin{\top}{}{0, \top}{\vlhy{}}
}
\)\\[1.25em]
\(
\id_{ \phi \otimes  \psi}
\dfn 
\vlderivation{
\vlin{\parr}{}{\phi^\perp\parr\psi^\perp, \phi \otimes  \psi}
{
\vliin{\otimes}{}{\phi^\perp,\psi^\perp, \phi \otimes  \psi}
{\vldr{\id_{\phi}
}{\phi^\perp, \phi}}
{\vldr{\id_{\psi}
}{\psi^\perp, \psi}}
}
}
\qquad 
\id _{\phi \oplus \psi} 
\dfn
\vlderivation{
\vliin{\with}{}{\phi^\perp \with \psi^\perp, \phi \oplus \psi}
{\vlin{\choice 0\oplus}{}{\phi^\perp, \phi \oplus \psi}{\vldr{\id_{\phi}
}{\phi^\perp, \phi}}}
{\vlin{\choice 1\oplus}{}{ \psi^\perp, \phi \oplus \psi}{\vldr{\id_{\psi}
}{\psi^\perp, \psi}}}
}
\)
\(
\id_{\mu X \phi}
\dfn 
\vlderivation{
\vlin{\nu}{}{\nu X \phi^\perp, \mu X \phi}{\vlin{\mu}{}{\phi^\perp(\nu X \phi^\perp), \mu X \phi}{
\vldr{
\id_{\phi(\mu X \phi)}
}{\phi^\perp(\nu X \phi^\perp),\phi(\mu X\phi)}
}
}
}
\)
    
    \caption{Coinductive definition of the identity derivation for the connectives \( \perp \), \( \unit \), \( \otimes \), \( \oplus \) and \( \mu \); the case of other connectives is defined dually.}
    \label{fig:id}
\end{figure}

\begin{rem}[Identity derivation]\label{rem:id}
  The sequent calculus did not include an identity rule. However, as observed in~\cite{Baelde2016InfinitaryPT}, we can  construct coinductively the \emph{identity derivation}  $\id_{\phi}$ as in \Cref{fig:id}. 
\end{rem}

\begin{prop}\label{fact:identity-branches} Every branch of  the identity derivation bears exactly two dual  threads. 
 In particular,      the identity derivation is (cut-free and) progressing.
\end{prop}
\begin{proof}
   By inspecting the construction of the identity derivation using \cref{lem:good-iff-bad}.  
\end{proof}

\subsection{Cut reduction  and $\omega$-normalisation}\label{sec:cut-reduction}

We recall that a cut in a derivation  is called \emph{critical} if both cut formulas are principal for a rule and \emph{non-critical} otherwise.

\begin{defn}[Cut reduction rules] 
	The \emph{critical cut reduction rules} (reducing critical cuts) and the \emph{commuting cut reduction rules} (reducing non-critical cuts) are defined in \cref{fig:cut-elim-prin,fig:cut-elim-comm}, respectively.\footnote{To avoid duplication of the cut reduction rules,  we will not assume an order of the premises of the cut.} Notice that the commuting rules allow for permutation of cut rules. As usual, we write  $\cutelims$ for reflexive and transitive closure of $\cutelim$.
 \end{defn}
 
    \begin{defn}[Reducts]
   Given a cut $\rrule$, the  cuts generated by applying a cut reduction step to $\rrule$ are called \emph{reducts (of $\rrule$)}. Notice that a cut can have zero or more reducts. 
\end{defn}

\begin{rem}\label{rem:at-least-one-reduct}
The only cut reduction step that has no reduct is the critical cut reduction step with cut formulas $\unit$ and $\perp$. 
\end{rem}

\begin{figure}
    \centering
   \[
 \vlderivation{
 \vliin{\cut}{}{\Gamma}
  {
  \vlin{\unit}{}{\unit}{\vlhy{}}
  }
  {
  \vlin{\bot}{}{\bot, \Gamma}{\vlhy{\Gamma}}
  }
 }
 \quad 
 \cutelim
 \qquad
 \vlderivation{
 \vlhy{\Gamma}
 }
   \]
   \[
   \vlderivation{
   \vliin{\cut}{}{\Gamma, \Delta}
    {
   \vlin{\mu}{}{\Gamma, \mu X\phi}{\vlhy{\Gamma, \phi(\mu X\phi)}}
   }
   {
\vlin{\nu}{}{\Delta, \nu X\phi^\perp}{\vlhy{\Delta, \phi^\perp(\nu X\phi^\perp)}}
}
}
\qquad 
\cutelim
\qquad 
\vliinf{\cut}{}{\Gamma, \Delta}{\Gamma, \phi(\mu X\phi)}{\Delta, \phi^\perp(\nu X\phi^\perp)}
   \]
   \[
   \vlderivation{
\vliin{\cut}{}{\Gamma, \Delta, \Sigma}
{
\vliin{\otimes}{}{\Gamma, \Delta, \phi \otimes \psi}{\vlhy{\Gamma, \phi}}{\vlhy{\Delta, \psi}}
}   
{
\vlin{\parr}{}{\phi^\perp\parr \psi^\perp, \Sigma}{\vlhy{\phi^\perp, \psi^\perp, \Sigma}}
}
   }
   \qquad \cutelim \qquad
   \vlderivation{
   \vliin{\cut}{}{\Gamma, \Delta, \Sigma}
    {\vlhy{\Delta, \psi}}
{\vliin{\cut}{}{\psi^\perp, \Gamma, \Sigma}{\vlhy{\Gamma, \phi}}{\vlhy{ \phi^\perp,\psi^\perp,  \Sigma}}
}
   }
   \]
   \[
\vlderivation{
\vliin{\cut}{}{\Gamma, \Delta}
{
\vlin{\oplus^i}{}{\Gamma, \phi_0 \oplus \phi_1}{\vlhy{\Gamma, \phi_i}}
}
{
\vliin{\with}{}{\phi_0^\perp \with \phi_1^\perp, \Delta}{\vlhy{\phi_0^\perp, \Delta}}{\vlhy{\phi_1^\perp, \Delta}}
}
}
\qquad
\cutelim
\qquad
\vlderivation{
\vliin{\cut}{}{\Gamma, \Delta}{\vlhy{\Gamma, \phi_i}}{\vlhy{\phi_i^\perp, \Delta}}
}
   \]
    \caption{Critical cut reduction rules.}
    \label{fig:cut-elim-prin}
\end{figure}
\begin{figure}
    \centering
    \[
    \vlderivation{
\vliin{\cut}{}{\Gamma, \Delta}{\vlin{\rrule}{}{\Gamma, \phi}{\vlhy{\Gamma_1, \phi}}}{\vlhy{\phi^\perp, \Delta}}
    }
    \qquad\cutelim\qquad
    \vlderivation{
    \vlin{\rrule}{}{\Gamma, \Delta}{\vliin{\cut}{}{\Gamma_1, \Delta}{\vlhy{\Gamma_1, \phi}}{\vlhy{\phi^\perp, \Delta}}}
    }
    \]
    \[
    \vlderivation{
    \vliin{\cut}{}{\Gamma, \Delta}{\vliin{\rrule\neq \with}{}{\Gamma, \phi}{\vlhy{\Gamma_1, \phi}}{\vlhy{\Gamma_2}}}{\vlhy{\phi^\perp, \Delta}}
    }
    \qquad \cutelim \qquad
    \vlderivation{
    \vliin{\rrule}{}{\Gamma, \Delta}{\vliin{\cut}{}{\Gamma_1, \Delta}{\vlhy{\Gamma_1, \phi}}{\vlhy{\phi^\perp, \Delta}}}{\vlhy{\Gamma_2}}
    }
    \]
     \[
    \vlderivation{
    \vliin{\cut}{}{\Gamma, \Delta}{\vliin{\with}{}{\Gamma, \phi, \psi \with \theta}{\vlhy{\Gamma, \phi, \psi}}{\vlhy{\Gamma, \phi, \theta}}}{\vlhy{\phi^\perp, \Delta}}
    }
    \qquad \cutelim \qquad
    \vlderivation{
    \vliin{\with}{}{\Gamma, \Delta, \psi \with \theta}{\vliin{\cut}{}{\Gamma, \Delta, \psi}{\vlhy{\Gamma_1, \phi, \psi}}{\vlhy{\phi^\perp, \Delta}}}{\vliin{\cut}{}{\Gamma, \Delta, \theta}{\vlhy{\Gamma, \phi, \theta}}{\vlhy{\phi^\perp, \Delta}}}
    }
    \]
    \caption{Commuting cut reduction rules. }
    \label{fig:cut-elim-comm}
\end{figure}

The notion of $\omega$-normalisation relies on the existence of $\omega$-long sequences of cut reduction steps - here called \emph{$\omega$-reduction sequences} -  that converge to a derivation.

\begin{defn}[$\omega$-reduction sequence] 
An \emph{$\omega$-reduction sequence (from $d$)} is an  $\omega$-indexed sequence of derivations $\reds s$ with $\reds s(0)=d$ and $\reds s(i) \cutelim^* \reds s({i+1})$ for all $i < \omega$. The sequence is called \emph{strict} if $\reds s (i)\neq \reds s (i+1)$ for infinitely many $i \in \omega$. \purple{We denote with $\depth{\reds s}i$ the minimal length of the (nodes labelled by) cut rules reduced by  $\reds s(i) \cutelims \reds s (i+1)$.} We say that $\reds s $ is \emph{depth-increasing} if either it is not strict or   $\lim_{i \in \omega} \depth{\reds s}i=\infty$. 
\end{defn}

Intuitively, non-strict $\omega$-reduction sequences implement \textit{finite} rewriting procedures. Moreover, the depth-increasing condition ensures the existence of \purple{the limit of an $\omega$-reduction sequence $\reds s$, which we denote by $\lim_{i \in \omega} \reds s(i)$}. Notice, however,  that such a limit may contain cuts (and  might not  satisfy progressivity). 

\begin{defn}[$\omega$-normalisation] A  derivation $d$ is \emph{$\omega$-normalisable} if there is a depth-increasing (possibly non-strict)  $\omega$-reduction sequence from $d$, and its limit is cut-free and progressing. We denote with $\Cnor{}$   the set of  $\omega$-normalisable derivations.
\end{defn}

\begin{rem}\label{rem:expansion}
       If $d \cutelim d'\in \Cnor{}$ then $d\in \Cnor{}$. 
\end{rem}

The following is a simple consequence of the $\omega$-compression property for transfinite reduction sequences of ill-founded  $\muMALL$ from~\cite{Saurin}. 

\begin{prop}[Compression]\label{prop:compression} Let $\reds s$ and $\reds t$ be depth-increasing $\omega$-reduction sequences from, respectively, $d$ and $\lim_{i \in \omega}\reds s(i)$. There exist a   depth-increasing $\omega$-reduction sequence $\reds r$ from $d$  such that   $\lim_{i \in \omega}\reds r(i)= \lim_{i \in \omega}\reds t$.
\end{prop}

\begin{prop}\label{prop:d-normalises-iff-pcut-eta-expansion-does}
Let $d$ be a derivation with conclusion $\phi_0, \ldots, \phi_n$. For all $i < n$:
\begin{enumerate}
\item there is a depth-increasing $\omega$-reduction sequence from $\cutform{d}{\id_{\phi_i}}{}$ whose limit is $d$.
\item There is a depth-increasing $\omega$-reduction sequence from $d$ 
    iff there is one from $\cutform{d}{\id_{\phi_i}}{}$ with the same limit.
\end{enumerate}
\end{prop}
\begin{proof}
The first claim is straightforward noticing that reducing the cut between $d$ and the  identity derivation does not affect the structure of $d$ while gradually pushing upward the cuts. In such a  sequence of cut reduction steps the cut rules in $d$ are not reduced and are merely commuted with the ``active'' cut.

Concerning point 2, let $\reds s$ be a depth-increasing $\omega$-reduction sequence from $d$ with limit  $d^*$. Clearly, we can construct a depth-increasing $\omega$-reduction sequence from $\cutform{d}{\id_{\phi_i}}{}$ with  limit $\cutform{d^*}{\id_{\phi_i}}{}$. Let $\reds t$ be this sequence. By point 1 there is a depth-increasing $\omega$-reduction sequence from $\cutform{d^*}{\id_{\phi_i}}{}$ and with  limit $d^*$. \Cref{prop:compression} yields a depth-increasing $\omega$-reduction sequence from  $\cutform{d}{\id_{\phi_i}}{}$ whose limit is $d^*$. 

  Concerning the converse of point 2, suppose there is $\reds s$ be a depth-increasing $\omega$-reduction sequence from $\cutform{d}{\id_{\phi_i}}{}$. We can clearly postpone all cut reduction steps applied to cuts in $d$, obtaining two  $\omega$-reduction sequences $\reds t$ and $\reds r$ such that:
  \begin{itemize}
      \item $\reds t$ is a depth-increasing $\omega$-reduction sequence from $\cutform{d}{\id_{\phi_i}}{}$ that  fully reduces the cut between $d$ and the identity derivation $\id_{\phi_i}$ (while permuting downward all the cuts in $d$). It has limit $d$.
      \item $\reds r$ is a   $\omega$-reduction sequence from $d$  that performs all the postponed reduction steps of $\reds s$ (with the same order as $\reds s$), and with the same limit as $\reds s$.  
  \end{itemize}
  
 Moreover, since $\lim_{i \in \omega}\depth{\reds s}i= \infty$ then $\lim_{i \in \omega} \depth{\reds r}i= \infty$. 
\end{proof}

\purple{
In the next section we introduce the notion of \textit{external progressivity}, a global condition that statically certifies $\omega$-normalisability, meaning that derivations satisfying this requirement can be easily shown to rewrite to a cut-free and progressing derivation in the limit (\Cref{lem:good-norm}).  
}

\section{Internal closure and external progressivity}\label{sec:good}

Infinitary cut elimination in the form of $\omega$-normalisation annihilates certain threads while preserving others. The two classes of threads can be characterised as the threads initiating from particular formula occurrences. The annihilated threads, which we refer to as \emph{internal}, are the threads stemming from cut formulas, while the preserved threads, called \emph{external}, are those initiating from the root of the derivation.

The next definitions clarify these notions.

\begin{defn}[Internal versus external]
Let $d$ be a derivation. A trace $t$ is \emph{internal} if $t(0)$ is a cut formula occurrence, and \emph{external} if $t(0)$ is a formula occurrence in the conclusion of $d$.  An \emph{internal} (resp. \emph{external}) thread of  a branch $\branch b$ is a thread of an  internal (resp. external) trace traversing $\branch b$. 
\end{defn}

Clearly, \emph{a trace is either internal or external}.

\begin{defn}[Coherence]
Let $d$ be a  derivation. For each thread $\tau$ we introduce a relation $\ortcoh{\tau}{}{}$ on branches of $d$ given by $\ortcoh{\tau}{\branch b}{\branch c}$ iff $\branch b \wedge \branch c$ is a cut in $d$ and there exists traces $s,t$ traversing $\branch b$ and $\branch c$ respectively such that:
\begin{itemize}
    \item $s(0)$ and $t(0)$ are the cut formula occurrences at nodes $\branch b(n+1)$ and $\branch c(n+1)$ respectively, where  $n= |\branch b \wedge \branch c|$
    \item $s$ bears $\tau$ and $t$ bears $\tau^\perp$.
 \end{itemize}
  \purple{In this case, we call $\branch{b}$ and $\branch{c}$ \emph{coherent}.} Notice that $\ortcoh{\tau}{\branch b}{\branch c}$ iff $\ortcoh{\tau^\perp}{\branch c}{\branch b}$.
\end{defn}

\emph{Internally closed sets} are closed sets of coherent branches (see preliminary conventions):

\begin{defn}[IC set] Let $d$ be a derivation. A non-empty set $X \subseteq \Branch{d}$ of branches of $d$
 is \emph{internally closed} (or \emph{IC sets}) if it is closed and, moreover,  for every $\branch b \in X$ and  $\tau $ internal thread of ${\branch b}$, there exists $\branch c \in X$ such that $ \ortcoh{\tau}{\branch b}{\branch c}$.
\end{defn}

IC sets allow us to introduce another global condition alternative to progressivity:

\begin{defn}[External progressivity]
    A  derivation is \emph{externally progressing} if every IC set bears a good external thread. We denote with $\GDer{}$ the set of derivations that are  externally progressing.
\end{defn}

\begin{rem}\label{rem:ext-prog-local-preserv}
    If $d \cutelim d'$ and $d\in \GDer{}$ then $d' \in \GDer{}$.
\end{rem}

\begin{figure}
    \centering
   \[
   \vlderivation{
\vliin{\cut}{}{\Gamma_0, \Delta_0}
{\vlin{\rrule_0}{}{\Gamma_0, \red{\nu X.X}}{\vlin{\rrule_1}{}{\Gamma_1, \red{\nu X.X}}{\vlin{\rrule_2}{}{\Gamma_2, \red{\nu X.X}}{\vlhy{\branch{b}\, \vdots}}}}}
{\vlin{\rrule'_0}{}{\blue{\mu X.X}, \Delta_0}{\vlin{\rrule'_1}{}{\blue{\mu X.X}, \Delta_1}{\vlin{\rrule'_2}{}{\blue{\mu X.X}, \Delta_2}{\vlhy{ \vdots\, \branch{b}'}}}}}
}
   \]
    \caption{IC sets and their cut elimination behaviour.}
    \label{fig:IC-set-example}
\end{figure}

An IC set $X$ represents a set of branches that are fully visited  by a cut reduction procedure. Specifically,   when  two branches $\branch b$ and $\branch {b}'$ are coherent, their cut  can  be eliminated only by  visiting \emph{both} branches in entirety, and  cut reduction will stepwise ``zip them up''  producing a cut-free limit branch. External progressivity  ensures that $X$ bears a good thread.

The following example clarifies the interplay between cut elimination and IC sets.

\begin{exmp}
  Consider the derivation in \cref{fig:IC-set-example} (left). It has only two branches, $\branch b$ and $\branch{b'}$ defined by $\branch{b}(n)= \langle 0, \overset{n}{\ldots}, 0 \rangle$ and $\branch{b}(n)= \langle 1, 0, \overset{n-1}{\ldots}, 0 \rangle$ respectively. \purple{Let $\tau$ (resp., $\tau'$) be the weak thread bourne by $\branch{b}$ (resp., $\branch{b}'$) and whose trace is highlighted with \red{red}  (resp., \blue{blue}) formula occurrences.} We have four cases:
 \begin{enumerate}
     \item Both $\tau$ and $\tau'$ are threads, and so $\tau$ is good and $\tau'=\tau^\perp$ is bad. By definition, the only IC set  is $\{\branch b, \branch{b}'\}$. Indeed, the cut can only be eliminated by performing infinitely many critical cut reduction steps, exploring both branches.
      \item \purple{$\tau$ is a thread but $\tau'$ is only a weak thread.  In this case, the only IC set is $\{\branch{b}'\}$. Indeed, no IC set can contain $\branch {b}$, as it bears an internal trace $\tau$ but $\branch{b}'$ does not bear $\tau^\perp$ ($\tau'$ is only a weak thread). To eliminate the cut, we need to explore entirely $\branch{b}'$ but we cannot do the same for $\branch{b}$, as that would require infinitely many critical cut reduction steps.}
      \item \purple{$\tau'$ is a thread but $\tau$ is only a weak thread. In this case, the only IC set is $\{\branch{b}\}$, and the only cut reduction process will visit entirely $\branch b$, but not $\branch{b}'$.}
    \item Neither $\branch b$ nor $\branch{b}'$ are infinitely often principal. We have three IC sets, namely $\{\branch{b}\}$, $\{\branch{b}'\}$, and $\{\branch b, \branch{b}'\}$. Indeed, to eliminate the cut we have three possible strategies, due to the fact that we will eventually perform only non-critical cut reduction steps: we can either eventually perform non-critical steps pushing the cut upward along only one branch (thus visiting entirely only one among $\branch{b}$ and $\branch{b}'$), or we can alternate non-critical steps that push upward the cut along both branches (thus fully visiting both of them).
    \end{enumerate}
\end{exmp}

\begin{rem}
Notice that not all the IC sets of a derivation can be fully visited during cut elimination. To see this, consider the derivation in \Cref{fig:IC-set-example}, and suppose that $\Gamma_i=\mygreen{\nu X.X}$ for all $i$, and that:
\begin{itemize}
    \item   $\red{\nu X.X}$ is principal for $\rrule_0$.
    \item  $\mygreen{\nu X.X}$ is principal for $\rrule_i$, for all   $i>0$. 
    \item no $\blue{\mu X.X}$ is principal for $\rrule'_i$.
\end{itemize}
 Clearly, $\{\branch b, \branch{b}'\}$ is a IC set, as neither $\branch b$ nor $ \branch{b}'$ bears an internal thread. However, since     $\red{\nu X.X}$ is principal for $\rrule_0$ and no $\blue{\mu X.X}$ is principal for $\rrule'_i$, there is only one cut elimination strategy, which repeatedly applies non-critical steps permuting downward all $\rrule'_i$. This means that, for all $n>1$,  $\branch{b}(n)$ will never become a premise of a cut. 
\end{rem}

One of the key results of this paper is that progressing derivations are also externally progressing (\Cref{lem:progressing-implies-reducible}). Notice that the converse does not hold, as the following example shows.

\begin{exmp}
  Let use consider the following derivation 
    \[
\vlderivation{
\vliin{\cut}{}{\red{\nu X .X}}{\vliq{\nu}{}{\red{\nu X.X}, \blue{\nu Y.Y}}{\vliq{\nu}{}{\red{X[\nu X .X/X]}, \blue{Y[\nu Y.Y/Y]}}{\vlhy{\vdots}}}}{\vlin{\mu}{}{\mygreen{\mu Y.Y}}{\vlin{\mu}{}{\mygreen{Y[\mu Y.Y/Y]}}{\vlhy{\vdots}}}}
}
    \]
It has only two branches, $\branch b$ and $\branch{b'}$ defined by  $\branch{b}(n)= \langle 0, \overset{n}{\ldots}, 0 \rangle$ and $\branch{b}(n)= \langle 1, 0, \overset{n-1}{\ldots}, 0 \rangle$, respectively.
\purple{Let $\tau$ (resp., $\sigma$) be the  thread bourne by $\branch{b}$  and whose trace is highlighted with \red{red}  (resp., \blue{blue})  formula occurrences. Moreover, let $\sigma^\perp$ be the  thread bourne by $\branch{b}'$ and whose trace is highlighted with \mygreen{green} formula occurrences.}
 \purple{Then, the only  IC set of the derivation is $\{\branch b, \branch {b'}\}$, as IC sets are non-empty and $ \ortcoh{\sigma}{\branch b}{\branch{b}'}$.} Moreover, $\{\branch b, \branch {b'}\}$ bears the good external thread, $\tau$, and so the derivation is externally progressing. However, the derivation is \emph{not} progressing, as the branch $\branch {b'}$ bears  no good thread.
\end{exmp}

However, the two global conditions  match in the cut-free setting.

\begin{prop}\label{prop:cut-free-good-equals-progressing}
A cut-free derivation $d$ is externally progressing iff $d$ is progressing.
\end{prop}
\begin{proof}
    By definition, as every thread is external in a cut-free derivation.
\end{proof}

\begin{rem}\label{rem:identity-is-good}
    From \cref{prop:cut-free-good-equals-progressing} and \cref{fact:identity-branches} it follows that the identity derivation is externally progressing.
\end{rem}

\section{Reducibility candidates}\label{sec:totality-candidates}

Reducibility candidates were introduced by Girard~\cite{Girard1972InterpretationFE} for proving normalisation of \purple{system} $\mathsf{F}$, a polymorphic calculus,  and generalise Tait's celebrated  \emph{reducibility} (or \emph{computability}) \emph{method} for   G\"{o}del's system $\mathsf{T}$ \cite{Tait67}. In the context of linear logic, reducibility candidates arise naturally from the so-called \textit{orthogonality construction} \cite{Girard87}, a recurring motif of this logic that stems from semantic motivations. Orthogonality-based reducibility candidates were developed by Baelde and Miller to prove normalisation of inductive $\muMALL$ in~\cite{BaeldeM07}. 

This section defines reducibility candidates for ill-founded $\muMALL$ via the notion of orthogonality, and outlines their properties. Our  constructions smoothly adapt~\cite{BaeldeM07}  to our infinitary setting. To recast the Baelde and Miller's normalisation argument for $\omega$-normalisation, we will introduce ordinal-indexed approximations of reducibility candidates.

\subsection{Orthogonality and reducibility candidates}

Henceforth, given a set of derivations $\Rset$  we will denote with $\Rset_\phi$ the set of derivations with a distinguished formula $\phi$ in the conclusion.

\begin{defn}[Orthogonality] 
	Let $\Rset$ be a set of derivations and \( \phi \) a formula. The ($\Rset$-)\emph{orthogonality relation} over \( \phi \)  is the relation \( {\perp_\phi^\Rset} \subseteq \Rset_\phi \times \Rset_{\phi^\perp} \) defined by $d\perp_\phi^{\Rset} d'$ iff $ \cutform d {d'} \phi\in \Rset$. 
	For an orthogonality relation \( {\perp} = {\perp_\phi^\Rset} \) and set of derivations $\mathcal{X} \subseteq \Rset_\phi $, we define \( \mathcal X ^{\perp} \coloneqq \{ d \in \Rset_{\phi^\perp} \mid e \perp d \text{ for all } e \in \mathcal X \}\).
\end{defn}

We write \( d \perp_\phi d' \) (or simply \( d \perp d' \)) if the set of derivations \( \Rset\) (and formula \(\phi\)) can be inferred from context. Thus, unless specified otherwise, when a set \( \mathcal{X} \) has been introduced in a context such as ``\( \mathcal{X} \subseteq \Rset_\phi \)'' the relation \( \perp \) denotes \( \perp_\phi^\Rset \). In particular, \( \mathcal{X}^{\perp \perp} \) means \( \mathcal{X}^{\perp_\phi \perp_{\phi^\perp}} \).

\begin{defn}[Reducibility candidate] 
	Let $\Rset$ be a set of derivations. A ($\Rset$-)\emph{reducibility candidate}  is a set $\mathcal{X}\subseteq \Rset_{\phi}$ satisfying $\mathcal{X}=\mathcal{X}^{\perp \perp}$.
\end{defn}

Intuitively, a $\Rset$-reducibility candidate  is a set of derivations that preserve membership in $\Rset$ when interacting by a cut rule with derivations from the dual candidates. In other words,  reducibility candidates carve out sets of derivations of $\Rset$ satisfying a ``composability property''. In this paper we will focus on $\Cnor{}$-reducibility candidates and $\GDer{}{}$-reducibility candidates.

In what follows, the class of derivations \( \Rset \) is considered fixed (and arbitrary) unless stated otherwise.
Orthogonality of derivations satisfies several  well-known properties which do not depend on the choice of $\Rset$.

\begin{prop}[See, e.g., \cite{BaeldeM07}]
    For any sets $\mathcal{X},\mathcal{Y} \subseteq \Rset$:
    \begin{itemize}
        \item  $\mathcal{X} \subseteq \mathcal{Y}$ implies $\mathcal{Y}^\perp\subseteq \mathcal{X} ^\perp $ 
        \item  $(\mathcal{X} \cup \mathcal{Y})^\perp = \mathcal{X}^\perp \cap \mathcal{Y}^\perp$
        \item $\mathcal{X} \subseteq \mathcal{X}^{\perp \perp}$
        \item $\mathcal{X}^{\perp\perp\perp}= \mathcal{X}^{\perp}$. In other words, $\mathcal{X}^\perp$ is always a reducibility candidate.
    \end{itemize}
\end{prop}

\begin{prop}
	 $\Rset$-reducibility candidates  (ordered by $\subseteq$) form a complete lattice.
\end{prop}
\begin{proof}
 Given an arbitrary collection of candidates $S$, it is easy to check that $(\bigcup S)^{\perp\perp}$ is its least upper bound in the lattice, and $\bigcap S$ its greatest lower bound. We just check the minimality of $(\bigcup S)^{\perp\perp}$: any upper bound $\mathcal{Y}$ satisfies $\bigcup S \subseteq \mathcal{Y}$, and hence $(\bigcup S)^{\perp\perp}\subseteq \mathcal{Y}^{\perp\perp}=\mathcal{Y}$.  Concerning the greatest lower bound, the only non-trivial thing is that it is a reducibility candidate, but it suffices to observe that $\bigcap S=\bigcap_{\mathcal{X} \in S} \mathcal{X}^{\perp\perp}=(\bigcup _{\mathcal{X} \in S}\mathcal{X}^\perp)^\perp$.  The least reducibility candidate is $\emptyset^{\perp\perp}$ and the greatest is $\Rset$. 
\end{proof}

Having a complete lattice, we can use the Knaster--Tarski theorem: any monotone function $f$ on reducibility candidates admits a least fixed point $\lfp f$ and a greatest fixed point $\gfp f$ in the lattice of candidates.

\subsection{Interpretation}

We now define an interpretation associating with every formula a reducibility candidate. Notice that the interpretation proceeds inductively on the construction of \emph{pre}formulas. 

\begin{defn}[Interpretation]\label{defn:interpretation}
 	An environment is a mapping $\env$ from the set of propositional variables to reducibility candidates. Let \( \env \) be an environment. By \( \env[X \mapsto \mathcal X] \) we denote the environment which maps \( X \) to \( \mathcal X \) and agrees with \( \env \) otherwise. We define, by recursion on a \emph{pre}formula $\phi$, a reducibility candidate $\inter \phi \env_\Rset$ (also denoted $\inter \phi \env$) called the \emph{interpretation of $\phi$ (under $\env$)}:
 \begin{itemize}
 	\item $\inter X \env_\Rset \dfn \env (X)$
 	\item $\inter\top \env_\Rset = \emptyset^{\perp}$
 	\item $\inter \perp \env_\Rset \dfn \left\{ \vlinf{\unit}{}{ \unit}{} \right\}^{\perp}$
 	\item $\inter {\phi \parr \psi} \env_\Rset\dfn \left\{ \begin{array}{c|c}
    \vlderivation{
 \vliin{\otimes}{}{\Delta, \Delta', \theta \otimes \chi}{\vldr{d}{\Delta, \theta}}{\vldr{d'}{\Delta', \chi}}    
}
     & d  \in \inter {\phi^\perp} \env_\Rset \text{ and } d'  \in \inter {\psi^\perp} \env_\Rset 
 	\end{array} \right\}^{\perp}$
 	\item $\inter {\phi_1 \with \phi_2}\env_\Rset \dfn \left\{\begin{array}{c|c} \vlderivation{\vlin{\oplus_i}{}{\Delta, \psi_1 \oplus \psi_2}{\vldr{d}{\Delta, \psi_i}}}  & \text{for }i \in \{0,1\} \text{ and } d  \in \inter {\phi_i^\perp}\env_\Rset  \end{array}  \right\}^{\perp} $
 		
	\item $\inter {\nu X \phi} \env_\Rset \dfn \displaystyle \bigcap \left\{ \mathcal X \quad \middle| \quad \mathcal X \subseteq  \left\{  \begin{array}{c|c}
    \vlderivation{
\vlin{\mu}{}{\Delta, \mu X \phi^\perp}{\vldr{d}{\Delta, \phi^\perp(\mu X \phi^\perp)} }    
}
& d \in \inter {\phi^\perp}{\env[ X \mapsto \mathcal{X}]}_\Rset
 		\end{array}  \right\}^{\perp}  \right\}$
 	\item In all other cases $\inter {\phi^\perp}\env_\Rset \dfn (\inter{\phi}{\env^\perp}_\Rset)^\perp $, where  $\env^\perp_\Rset: X \mapsto \env_\Rset(X)^\perp$.
 \end{itemize}
 Finally, for \( \Gamma = \phi_1, \ldots, \phi_n \) we define $\inter{\Gamma}{\env}_\Rset = \inter{\phi_1}{\env}_\Rset \cap \dotsm \cap \inter{\phi_n}\env_\Rset$, which is a reducibility candidate by~\Cref{prop:properties-candidates}. Henceforth, we omit the environment $\env$ when there are no free variables or the choice is clear from the context.
\end{defn}

The following properties of interpretations are easy consequences of the definition, 
 recalling that $X=X^\perp$ from~\Cref{subsec:formulas}.

\begin{prop}
    \label{prop:candidates-subset-R} For all \emph{formulas} $\phi$, \( \inter{\phi}{} \subseteq \Rset_\phi \).
\end{prop}
 
\begin{prop}\label{prop:properties-candidates}
	For all pre-formulas \( \phi \) and $\psi$:
    \begin{enumerate}
        \item   \( \inter {\phi^\perp} \env = (\inter {\phi}{\env^\perp} )^\perp \).
        \item \label{inter-subst} $\inter{\phi[\psi/X]} \env = \inter{\phi}{\env[X \mapsto \inter \psi {\env}]} $.
    \end{enumerate}
\end{prop}

\begin{lem}
	[Monotonicity] 
	\label{lem:mon-pos-types-in-hr}
	If $\mathcal{X}\subseteq \mathcal{Y}$ are reducibility candidates and \( \phi \) is a pre-formula then \[ \inter{\phi}{\env[ X \mapsto \mathcal{X}]}\subseteq \inter{\phi}{\env[X \mapsto \mathcal{Y}]}. \] 
In particular, 	$\mathcal{X} \mapsto \inter{\phi}{\env[ X \mapsto \mathcal{X}]}$ is a monotone endofunction over the complete lattice of the reducibility candidates.
\end{lem}

It is important to note that fixed points are not interpreted as bona fide fixed points, i.e., $\inter {\mu X \phi} \env\neq \inter {\phi(\mu X \phi)}{\env}$ in general. Rather, we have:
\begin{prop}\label{prop:characterisation-of-interpretation-of-mu}
$\inter {\nu X \phi} \env\dfn \left\{  \begin{array}{c|c}
				\vlderivation{
					\vlin{\mu}{}{\Delta, \mu X. \phi^\perp}{\vldr{d}{\Delta, \phi^\perp(\mu X \phi^\perp)} }    
				}
				& d \in \inter {\phi^\perp(\mu X \phi^\perp)}{\env}
			\end{array}  \right\}^{\perp}.$
\end{prop}
\begin{proof}
	It follows from the Knaster-Tarski theorem, using~\Cref{lem:mon-pos-types-in-hr} and \Cref{prop:properties-candidates}.\ref{inter-subst}.
\end{proof}

\subsection{Ordinal assignments}

Our cut elimination arguments based on reducibility candidates will rely on an alternative characterisation of fixed points through ordinal-indexed approximations. To this end, we temporarily extend the formal syntax:
\begin{itemize}
	\item for any formulas $\mu X\phi$ and $\nu X\phi$, and for any ordinal $\alpha$, the expressions $\mu ^\alpha X \phi$ and $\nu^\alpha X \phi$ are also formulas.
\end{itemize}
Let us stress that indexed quantifiers are only permitted in the case of \emph{formulas}, so whereas \( \nu^\omega X \nu Y .\, X \oplus Y \) is permitted, \( \nu X \nu^\omega Y .\, X \oplus Y \) is not because \( \nu Y.\, X \oplus Y \) is not a formula.

The interpretation function \( \inter \cdot {\env}_\Rset \) is duly expanded by the following condition for all \( \alpha \).
\begin{itemize}
    \item $\inter{\nu^\alpha X\phi}{\env} \dfn \displaystyle \bigcap_{\beta < \alpha} \left\{  \begin{array}{c|c}
				\vlderivation{
					\vlin{\mu}{}{\Delta, \mu X. \phi^\perp}{\vldr{d}{\Delta, \phi^\perp(\mu X. \phi^\perp)} }    
				}
				&  d \in \inter{\phi^\perp}{\env[X \mapsto \inter{\nu^\beta X \phi}{\env}]} 
			\end{array}  \right\}^{\perp}$
    \item $\inter{\mu^\alpha X\phi}{\env} \dfn 
	(\inter{\nu^\alpha X \phi^\perp}{\env^\perp})^\perp$
\end{itemize}
notice that $\inter{\nu^0 X\phi}{\env} = \Rset$. By the Knaster--Tarski theorem:
\begin{prop}
	\label{cor:hr-mu-as-limit-of-approximants} 
	Let $\Rset$ be a set of derivations and let \( \Ord \) denote the class of all ordinals. For any pre-formula \( \phi \) with at most \( X \) free:
    \begin{enumerate}
        \item $\inter{\mu X \phi}{} = \displaystyle\bigcup\limits_{\alpha \in \Ord} \inter{\mu^\alpha X \phi}{}$.
        \item  $\inter{\nu X \phi}{} = \displaystyle\bigcap\limits_{\alpha \in \Ord} \inter{\nu^\alpha X \phi }{}$.
    \end{enumerate}
\end{prop}

Given a formula $\phi$ and a list of ordinals of appropriate length $\vec\alpha$, we can assign ordinals  to the $\nu$-subformulas of $\phi$ in such a way that the interpretation of the resulting formula is an approximation of the interpretation of $\phi$, that is, $\inter{\phi^{\vec \alpha}}{}$ approximates $\inter{\phi}{}$. A similar reasoning applies to the $\mu$-subformulas of $\phi$. This requires the following subtle definition of assignment.

\begin{defn}[Assignments]
	Let \( \phi \) be a formula and let \( \nu X_k \phi_k < \dots < \nu X_1 \phi_1 \) enumerate the \( \nu \)-formulas that are \(\fleq\)-equivalent to \( \phi \) in order of priority. Given a sequence of ordinals \( \vec \alpha = \alpha_1 \dotsm \alpha_l \) with \( l \leq k \), referred to as a \( \nu \)\emph{-assignment}, we define the formula \( \phi^{\vec \alpha} \) by recursively replacing each subformula \( \nu X_i \phi_i \) by \( \nu^{\alpha_i} X_i \phi_i\), starting at \( i = l \).

	The formula \( \phi_{\vec \alpha} \) is defined analogously with respect to the \( \mu \)-subformulas of $\phi$, in which case \( \vec \alpha \) is referred to as a \( \mu \)\emph{-assignment}.
\end{defn}

Assignments are assumed to be ordered lexicographically, i.e., $\vec \alpha < \vec \beta$ iff \( \vec \alpha \) is a proper prefix of \( \vec \beta \) or there exists  $i$ with $\alpha_i < \beta_i$ and $\alpha_j = \beta_j$ for all $j<i$. The empty sequence is considered the least ordinal sequence. Note, the lexicographic ordering is well-founded on sequences of ordinals of bounded length. 
By \cref{lem:mon-pos-types-in-hr} and \Cref{cor:hr-mu-as-limit-of-approximants} we have:
\begin{prop}
	[Positive and negative approximants]
	\label{prop:pos-neg-approximants} {\ }
\begin{itemize}
    \item If $d \in \inter{\psi}{}$ there are (least) ordinals $\vec \alpha$ s.t.\ $d \in \inter{\psi_{\vec \alpha}}{}$.
    \item If $d \notin \inter \psi{}$ there are (least) ordinals $\vec \alpha$ s.t.\ $d \notin \inter{\psi^{\vec \alpha}}{}$.
\end{itemize}
\end{prop}

\subsection{Reducibility}

We now introduce  \emph{reducibility}. Intuitively, a derivation is $\Rset$-reducible if it can be turned to a derivation of $\Rset$ by cutting its conclusions with derivations from (dual)  reducibility candidates. A key result is that, when the set $\Rset$ satisfies some appropriate properties, as in the case of  $\Cnor{}$ and  $\GDer{}$, the $\Rset$-reducible derivations with conclusion $\Gamma$ will be precisely those belonging to the reducibility candidate $\inter{\Gamma}{}_\Rset$.  Because this characterisation result  does not hold for general sets $\Rset$, many results of this subsections will be restricted to the relevant cases of $\Cnor{}$ and $\GDer{}$.

\begin{defn}[$\Rset$-reducibility]
	A derivation \( d \) with conclusion  $\phi_0, \ldots, \phi_n$ is 
	(\( \Rset \)-)\emph{reducible} if $\cutform{d}{d_0, \ldots, d_n}{}\in \Rset$ for all derivations $(d_i \in \inter{\phi_i^\perp}{}_\Rset)_{i \le n}$.
\end{defn}

\begin{prop}\label{prop:equivalent-reducibility}
    Let $\Rset \in \{ \GDer{}, \Cnor{}\}$ and $d$ be a derivation with conclusion $  \phi_0 , \dotsc, \phi_n $.  The following are equivalent for all \( j \le n \):
    \begin{enumerate}
        \item \label{prop:equiv-red-1} $d$ is reducible.
        \item \label{prop:equiv-red-2} for all $(d_i\in \inter{\phi_i^\perp}{})_{i \neq  j}$, $\cutform{d}{d_0, \dotsc, d_{j-1}, d_{j+1}, \dotsc d_{n}}{}\in \inter{\phi_j}{}$.
    \end{enumerate}
\end{prop}
\begin{proof}

We treat the case \( j = n \) for simplicity.
Suppose \( d \) is reducible. Since $\inter{\phi_n}{}= {\inter{\phi_n}{}}^{\perp \perp}$, it suffices show that the derivation $d'=\cutform{d}{d_0 , \dotsc, d_{n-1}}{}$ is orthogonal to ${\inter{\phi_n}{}}^{\perp}$. 
Let $d_n \in {\inter{\phi_n}{}}^\perp$. Using cut permutation reductions, \( \cutform{d'}{d_n}{} \) reduces to $\cutform{d}{d_0, \ldots, d_n}{}$, which is  in $\Rset$ by the fact that $d$ is reducible and $\Rset$ is closed under permutation of cut rules.

That \ref{prop:equiv-red-2} implies \ref{prop:equiv-red-1} can be realised in a similar way. We have to show that $\cutform{d}{d_1, \ldots, d_n}{}\in \Rset$ for $(d_i \in \inter{\phi_i^\perp}{})_{i \le n}$. Permuting the cuts yields  $\cutform{\cutform{d}{d_0 , \dotsc, d_{n-1}}{}}{d_n}{}$, which is in $\Rset{}$ because $d_n \in \inter{\phi_n^\perp}{}={\inter{\phi_n}{}}^\perp$ and $\cutform{d}{d_0, \dotsc, d_{n-1}}{}\in \inter{\phi_n}{}$ by hypothesis.
\end{proof}

The characterisation result of $\Rset$-reducibility  relies on the following properties of the identity derivation.

\begin{prop}\label{prop:id-in-tot-candidate} Let $\Rset\in \{\GDer{}, \Cnor{}\}$. For every formula \( \phi \),
\begin{enumerate}
\item \label{enum:id-cond-1} $d \in \Rset$ iff $\cutform{d}{\id_{\phi}}{\phi}\in \Rset$
    \item \label{enum:id-conds-2} $\id_{\phi}\in \inter{\phi^\perp, \phi}{}$.
    \item \label{enum:id-conds-3} $d \in \inter{\phi}{}$ iff $\cutform{d}{\id_{\phi}}{\phi}\in \inter{\phi}{}$.
\end{enumerate}
\end{prop}
\begin{proof}
 Let us show point~\ref{enum:id-cond-1}.  If $\Rset=\Cnor{}$ then we apply \cref{prop:d-normalises-iff-pcut-eta-expansion-does}. \purple{The case
 $\Rset=\GDer{}$ follows by noticing that every branch of the identity derivation has exactly two dual threads,  one is good and the other one is bad by~\Cref{lem:good-iff-bad}.} 
 Point \ref{enum:id-conds-2} follows directly from point \ref{enum:id-cond-1}. As for point \ref{enum:id-conds-3},  the left-to-right direction follows from point \ref{enum:id-conds-2}. Concerning the converse, suppose $\cutform{d}{\id_{\phi}}{\phi}\in \inter{\phi}{}$, and let   
  $e\in \inter{\phi^\perp}{}$. By assumption,   $\cutform{\cutform{d}{\id_{\phi}}{\phi}}{e}{\phi}\in \Rset$, and so $\cutform{d}{\cutform{\id_{\phi}}{e}{\phi}}{\phi}\in \Rset$ by commuting cuts. By \cref{prop:d-normalises-iff-pcut-eta-expansion-does}, we have $\cutform{d}{e}{\phi}\in \Rset$. Therefore $d \in \inter{\phi}{}$.
\end{proof}

\begin{prop}[Characterisation]\label{prop:reducible-implies-in-interpretation}
 Let $\Rset \in \{\GDer{}, \Cnor{}\}$ and let $d$ be a derivation with conclusion $\Gamma$. Then $d$ is $\Rset$-reducible  iff  $d \in \inter{\Gamma}{}_\Rset$.
\end{prop}
\begin{proof}
   Let $\Gamma=\phi_0, \ldots, \phi_n$. The right-to-left direction follows by definition. Concerning the converse, suppose  $d$ is reducible. We show that $d \in \inter{\phi_i}{}$ for each $i$. By \cref{prop:id-in-tot-candidate}, $\id_{\phi_j}\in \inter{\phi_j^\perp}{}$ for every \( j \), so  \Cref{prop:equivalent-reducibility} implies 
   $ \cutform{d}{\id_{\phi_0}, \dotsc, \id_{\phi_{i-1}}, \id_{\phi_{i+1}}, \dotsc \id_{\phi_{n}}}{}\in \inter{\phi_i}{} 
   $. 
  Repeatedly applying \cref{prop:id-in-tot-candidate} yields $d\in \inter{\phi_i}{}$. Therefore $d \in \inter{\Gamma}{}$.
\end{proof}

It is convenient to expand the concept of {reducibility} to formulas endowed with ordinal assignments. 
\begin{defn}[$\Rset$-reducibility, revisited]
	Let \( \vec \alpha_0 \), \dots, \( \vec \alpha_n \) be sequences of ordinals. 
A derivation with conclusion \(\phi_0, \dotsc, \phi_n \)  is \( (\vec \alpha_0, \dotsc, \vec \alpha_n) \)\emph{-reducible (to \( \Rset \))} iff \( \cutform{d}{e_0, \dotsc, e_n}{} \in \Rset \) for all derivations \( (e_i \in \inter{\phi_i^{\vec \alpha_i}}{\perp}_\Rset)_{i \le n} \).
\end{defn}

\begin{prop}[\Cref{prop:equivalent-reducibility}, revisited]\label{reduc-approx-charac}
Let $\Rset \in \{\Cnor{}, \GDer{}\}$	and let  \( d \) be a derivation with conclusion \( \phi_0 , \dotsc, \phi_n \). The following are equivalent.
	\begin{enumerate}
		\item \( d \) is \( (\vec \alpha_0, \dotsc, \vec \alpha_n) \)-reducible,
		\item \( \cutform{d}{e_0, \dotsc, e_{n-1}}{} \in \inter{\phi_n^{\vec \alpha_n}}{} \) for all \( e_j \in \inter{\phi_j^{\vec \alpha_i}}\perp \) (\( j < n \)).
	\end{enumerate}
\end{prop}

The following is an immediate consequence of \cref{prop:pos-neg-approximants}.

\begin{prop}\label{prop:reducibility-approximations}
	Let $\Rset \in \{\Cnor{}, \GDer{}\}$ and  \( d \) be a derivation with conclusion $\phi_0, \ldots, \phi_n$. 
    If $d$ is not reducible then there exist ordinal sequences $\vec \alpha_0, \ldots, \vec \alpha_n$  such that
    \begin{enumerate}
    	\item \( d \) is not $(\vec \alpha_0, \dotsc, \vec \alpha_n)$-reducible and
    	\item If \( \vec \beta_i \le \vec \alpha_i \) for all \( i \le n \) and \( \vec \beta_j < \vec \alpha_j \) for some \( j \le n \) then \( d \) is $(\vec \beta_0, \dotsc, \vec \beta_n)$-reducible.
    \end{enumerate}
\end{prop}

\section{First cut elimination argument}\label{sec:cut-elimination-theorem}

Our first cut elimination proof is based on $\Cnor{}$-reducibility candidates. We will show that every progressing derivation with conclusion $\Gamma$ belongs to $\inter{\Gamma}{}_{\Rset}$ (for $\Rset\in \{\Cnor{}, \GDer{}\}$), and therefore is $\omega$-normalising by~\Cref{prop:candidates-subset-R}. Our approach follows a standard argument for establishing logical soundness (or computational totality) of ill-founded proof systems (see, e.g., \cite{muLJ,farzad-infinitary}). First, we  show that the inference rules of the system are sound for  our  reducibility candidates   semantics,  both w.r.t.  $\Cnor{}$-reducibility and $\GDer{}$-reducibility (\Cref{lem:reflecting-non-totality}). Then, we assume towards contradiction  $d \not \in \inter{\Gamma}{}_{\Cnor{}}$, from which we infer that $d$ does not belong to some approximation of the latter candidate (\Cref{prop:pos-neg-approximants}). By  repeatedly appealing to (the contrapositive of) \Cref{lem:reflecting-non-totality} we   construct a branch of $d$ reflecting this non-membership property. Finally, by progressivity of $d$ we can find a good thread along this branch decreasing infinitely often such ordinal approximation, contradicting well-foundedness of ordinals.

 \begin{lem}[Local soundness]\label{lem:reflecting-non-totality}
\purple{Let $\Rset\in\{\GDer{}, \Cnor{}\}$, and let $d$ have the form
$$
	\vlderivation{
	\vliiin{\rrule}{}{\theta_0, \ldots, \theta_n}{\vldr{d_0}{\Gamma_0}}{\vlhy{\ldots}}{\vldr{d_k}{\Gamma_k}}
	}
$$
If $d$ is not reducible  then there is $i \le k$ such that 
\begin{enumerate}
	\item $d_{i}$ is not reducible.
	\item For $\Gamma_{i} = \xi_0, \ldots, \xi_m$, if \( d \) is not $(\vec\alpha_0, \ldots, \vec  \alpha_n)$-reducible  then there exist ordinals sequences $\vec \beta_0, \ldots, \vec \beta_m$ such that \( d_i \) is not $(\vec \beta_0, \ldots, \vec \beta_m)$-reducible  and, for all \( j \le n  \) and \( l \le m \),
	\begin{itemize}
		\item if $\xi_{l}$ is an ancestor of $\theta_j$ and \( \theta_j \fleq \xi_{l}\) then $\vec \beta_{l} \leq \vec \alpha_j$
		\item If $\rrule=\nu$,  then $m=n$ and $\vec \beta_{m}< \vec \alpha_{n} $.
	\end{itemize}
\end{enumerate}
}
\end{lem}
\begin{proof}
 In all cases except \( \rrule = \nu \), the second claim is a trivial strengthening of the first argument by \cref{prop:reducibility-approximations},  and therefore omitted.   We proceed by case analysis on the rule \( \rrule \).
 Let    
 \( \Delta = \theta_0, \ldots, \theta_{n-1} \).
\begin{itemize}
\item $\rrule=\cut$. So $k=1$. Let \( \phi \) be the cut formula, that is \( \Gamma_0 = \Delta_0 , \phi \) and \( \Gamma_1 = \Delta_1 , \phi^\perp \) with \( \Gamma = \Delta_0 , \Delta_1 \). Suppose that $d_0$ is reducible and  $d_1$ is reducible. 
Fix $e_j \in \inter{\theta_j^\perp}{}={\inter{\theta_j}{}}^\perp$ for \( j \le n \).
\Cref{prop:equivalent-reducibility} implies 
\begin{align*}
	\cutform{d_0}{e_0, \ldots, e_{m-1}}{} &\in \inter{\phi}{}
	\\
	\cutform{d_1}{e_{m}, \ldots, e_{n}}{} &\in \inter{\phi^\perp}{}={\inter{\phi}{}}^\perp
\end{align*}
where \( m = \lvert \Delta_0 \rvert \).
Thus, 
$$
\cutform{\cutform{d_0}{e_0, \ldots, e_{m-1}}{}}{\cutform{d_1}{e_{m}, \ldots, e_{n}}{}}{} \in \Rset{}. 
$$
By~\Cref{rem:expansion} and the fact that $\GDer{}$ is invariant under permuting cuts,  we infer
$$
\cutform{\cutform{d_0}{d_1}{}}{e_0, \ldots, e_{n}{}}{} \in \Rset.
$$
So \( d \) is reducible.

\item $\rrule= \parr$. So $k=0$, $\theta_n=\phi\parr \psi$ and $\Gamma_0 = \Delta, \phi, \psi$. Suppose that $d$ is not reducible. Appealing to \cref{prop:equivalent-reducibility}, let $ e_j \in \inter{\theta_j^\perp}{}={\inter{\theta_j}{}}^\perp$ for \( j < n \) be such that 
$$ \cutform{d}{e_0, \ldots, e_{n-1}}{}\not \in \inter{\phi \parr \psi}{}.
$$
Let the conclusion of \( e_j \) be \( \theta_j^\perp, \Theta_j \). Recall that
$$
\inter{\phi \parr \psi}{}=
\left\{ \begin{array}{c|c}
    \vlderivation{
	 \vliin{\otimes}{}{\Sigma', \Sigma'', \phi^\perp \otimes \psi^\perp}{\vldr{f}{\Sigma', \phi^\perp}}{\vldr{g}{\Sigma'', \psi^\perp}}    
	}
	 & f  \in \inter {\phi^\perp} {}\ ,\ g  \in \inter{\psi^\perp} {}
	\end{array} \right\}^{\perp}.
$$
So, there exist derivations $f \in \inter{\phi^\perp}{}={\inter{\phi}{}}^\perp$ and $g \in \inter{\psi^\perp}{}={\inter{\psi}{}}^\perp$ such that the following derivation is not in $\Rset{}$:
$$
\vlderivation{
\vliin{\cut}{}{\Theta_0, \ldots, \Theta_{n-1}, \Sigma, \Sigma'}
{
\vliiiiq{\cut}{}{\Theta_0, \ldots, \Theta_{n-1}, \phi \parr \psi}{
\vlin{\parr}{}{\Delta, \phi \parr \psi}{\vldr{d_0}{\Delta, \phi, \psi}}
}
{
\vldr{e_0}{\theta^\perp_0, \Theta_0}
}
{\vlhy{\ldots}}
{
\vldr{e_{n-1}}{\theta^\perp_{n-1}, \Theta_{n-1}}
}
}
{\vliin{\otimes}{}{\Sigma, \Sigma', \phi^\perp \otimes \psi^\perp}{\vldr{f}{\Sigma, \phi^\perp}}{\vldr{g}{\Sigma,' \psi^\perp}}  }
}
$$
We apply a series of cut elimination steps and obtain:
$$
\vlderivation{
\vliin{\cut}{}{\Theta_0, \ldots, \Theta_{n-1}, \Sigma, \Sigma'}
{
\vliin{\cut}{}{\Theta_0, \ldots, \Theta_{n-1}, \Sigma, \Sigma', \psi}
{
\vliiiiq{\cut}{}{\Theta_0, \ldots, \Theta_{n-1}, \phi , \psi}{
\vldr{d_0}{\Delta, \phi, \psi}
}
{
\vldr{e_0}{\theta^\perp_0, \Theta_0}
}
{\vlhy{\ldots}}
{
\vldr{e_{n-1}}{\theta^\perp_{n-1}, \Theta_{n-1}}
}
}
{\vldr{f}{\Sigma, \phi^\perp}}
}
{\vldr{g}{\Sigma', \psi^\perp}}}
$$
By~\Cref{rem:expansion} and inspecting the definition of $\GDer{}$ we have that the above derivation is not in $\Rset$. Therefore,  $d_0$ is not reducible.

\item $\rrule= \with$. So $k=1$, $\theta_{n}= \phi_0\with \phi_1$, $\Gamma_0= \Delta, \phi_0$, and $\Gamma_1= \Delta, \phi_1$. {This case is delicate for $\GDer{}$-reducibility, as the critical cut elimination step for the rule $\with$ discards one of its premises,  erasing a subderivation.  So, some relevant information about threads might disappear after that cut elimination step. For this reason, we will treat  $\Cnor{}{}$-reducibility and $\GDer{}$-reducibility differently for this case.  }

 Suppose $d$ is not $\Rset{}$-reducible, so there exist derivations $e_j\in \inter{\theta_j^\perp}{}={\inter{\theta_j}{}}^\perp$ for \( j < n \) such that 
$$ \cutform{d}{e_0, \ldots, e_{n-1}}{}\not \in \inter{\phi_0 \with \phi_1 }{}.
$$
By the definition of \( \inter{\phi_0 \with \phi_1 }{} \) this means that there exists \( i \in \{ 0,1\} \) and $e \in \inter{\phi_i^{\perp}}{}={\inter{\phi_i}{}}^\perp$ such that the  following derivation is not in $\Rset{}$:
$$
d^*_1=
\vlderivation{
\vliin{\cut}{}{\Theta_0, \ldots, \Theta_{n-1}, \Sigma}
{
\vliiiiq{\cut}{}{\Theta_0, \ldots, \Theta_{n-1}, \phi_0 \with \phi_1}{
\vliin{\with}{}{\Delta, \phi_0 \with \phi_1}{\vldr{d_0}{\Delta, \phi_0}}{\vldr{d_1}{\Delta, \phi_1}}
}
{
\vldr{e_0}{\theta^\perp_0, \Theta_0}
}
{\vlhy{\ldots}}
{
\vldr{e_{n-1}}{\theta^\perp_{n-1}, \Theta_{n-1}}
}
}
{\vlin{\oplus}{}{ \Sigma, \phi_0^\perp \oplus \phi_1^\perp}{\vldr{e}{\Sigma, \phi_i^\perp}}  }
}
$$
for some  $\Theta_1, \ldots, \Theta_{n-1}, \Sigma$. Applying a series of cut elimination steps yields the following derivation:
$$
d^*_2=
\vlderivation{
\vliin{\cut}{}{\Theta_0, \ldots, \Theta_{n-1},  \Sigma}
{
\vliiiiq{\cut}{}{\Theta_0, \ldots, \Theta_{n-1}, \phi_i}{
\vldr{d_i}{\Delta, \phi_i}
}
{
\vldr{e_0}{\theta^\perp_0, \Theta_0}
}
{\vlhy{\ldots}}
{
\vldr{e_{n-1}}{\theta^\perp_{n-1}, \Theta_{n-1}}
}
}
{\vldr{e}{\Sigma, \phi_i^\perp}}
}
$$
 By~\Cref{rem:expansion} we infer that the $d^*_2$ is not in $\Cnor{}$. Therefore, \( d_i \) is not $\Cnor{}$-reducible. Concerning, $\Rset=\GDer{}$,  we have two cases. If $d_{1-i}$ is not $\GDer{}$-reducible, then we are done. Otherwise, it is  $\GDer{}$-reducible, and by~\Cref{prop:reducible-implies-in-interpretation} and~\Cref{prop:candidates-subset-R}, $d_{1-i}\in \GDer{}$. Now, let \( X \) be an IC set in $d^*_1$ that does not bear any good external thread, \purple{and let $Y$ be obtained from $X$ by discarding the branches that traverse $d_{1-i}$. Since $d_{1-i}\in \GDer{}$ then $Y \neq \emptyset$. Moreover, $Y$ is clearly an IC set. But since $X$ does not bear a good external thread, then neither does $Y$. With minor adaptations, $Y$ can be turned into an IC set of $d^*_2$, which implies  $d^*_2\not \in \GDer{}$.}  Therefore, \( d_i \) is not $\GDer{}{}$-reducible.

\item $\rrule= {\perp}$. So $k=0$ and \( \theta_n= \bot \). Suppose $d$ is not reducible,  so there exist derivations $e_j\in \inter{\theta_j^\perp}{}=\inter{\theta_j}\perp$ for \( j< n \) such that $$ \cutform{d}{e_0, \ldots, e_{n-1}}{}\not \in \inter{\perp}{}= \left\{ \vlinf{\unit}{}{ \unit}{} \right\}^{\perp}. $$ In other words, the  following derivation is not in $\Rset{}$:
$$
\vlderivation{
\vliin{\cut}{}{\Theta_0, \ldots, \Theta_{n-1}}
{
\vliiiiq{\cut}{}{\Theta_0, \ldots, \Theta_{n-1}, \perp}
{\vlin{\perp}{}{\Gamma_0, \perp}{\vldr{d_0}{\Gamma_0}}}
{\vldr{ e_0}{\theta_0^\perp, \Theta_{n-1}}}
{\vlhy{\ldots}}
{\vldr{ e_{n-1}}{\theta_{n-1}^\perp, \Theta_{n-1}}}
}
{
\vlin{\unit}{}{\unit}{\vlhy{}}
}
}
$$
for appropriate $\Theta_0, \ldots, \Theta_{n-1}$. By applying a series of cut elimination steps, \Cref{rem:expansion} and inspecting  the definition of $\GDer{}$, we have that \( \cutform{d_0}{e_0, \dotsc, e_{n-1}}{} \not\in \Rset{} \). Therefore, $d_0$ is not reducible.

\item $\rrule\in \{\top, \unit, \ex, \otimes, \oplus, \mu \}$. The case of \( \top \) and \( \unit \) hold vacuously. Exchange is straightforward and the cases \( \otimes \) and \( \oplus \) follow immediately from the definition of the interpretation. The case of \( \mu \) is similar but via \Cref{prop:characterisation-of-interpretation-of-mu}.
\item
$\rrule= \nu$. So $k=0$, $\theta_{n}= \nu X\phi$ for some \( X \) and \( \phi \), and $\Gamma_0 = \Delta , \phi(\nu X\phi)$. Suppose that $d$ is not reducible. Let $e_j\in \inter{\theta_j^\perp}{}={\inter{\theta_j}{}}^\perp$ for \( i < n \) be such that $$ \cutform{d}{e_0, \ldots, e_{n-1}}{}\not \in \inter{\nu X\phi}{}. $$
\Cref{prop:characterisation-of-interpretation-of-mu} established
$$
\inter {\nu X \phi}{}_{\GDer{}} = 
\left\{  \begin{array}{c|c}
			\vlderivation{
				\vlin{\mu}{}{\Sigma, \mu X \phi^\perp}{\vldr{e}{\Sigma, \phi^\perp(\mu X \phi^\perp)} }    
			}
			& e \in \inter{\phi^\perp(\mu X \phi^\perp)}{}
\end{array}  \right\}^{\perp}
$$
so there exists some derivation $e \in \inter {\phi^\perp(\mu X \phi^\perp)}{}$ such that the following derivation is not in $\Rset{}$:
\begin{gather}\label{eqn:nu-reduc-approx}
\vlderivation{
\vliin{\cut}{}{\Theta_0, \ldots, \Theta_{n-1}, \Sigma}
{
\vliiiiq{\cut}{}{\Theta_0, \ldots, \Theta_{n-1}, \nu X.\phi}{
\vlin{\nu}{}{\Delta, \nu X \phi}{\vldr{d_0}{\Delta, \phi(\nu X\phi)}}
}
{
\vldr{e_0}{\theta^\perp_0, \Theta_0}
}
{\vlhy{\ldots}}
{
\vldr{e_{n-1}}{\theta^\perp_{n-1}, \Theta_{n-1}}
}
}
{\vlin{\mu}{}{\Sigma , \mu X \phi^\perp}{\vldr{e}{\Sigma, \phi^\perp(\nu X \phi^\perp)}}  }
}
\end{gather}
for suitable $\Theta_1, \ldots, \Theta_{n-1}, \Sigma$. 
We apply a sequence of cut elimination steps  we obtain:
$$
\vlderivation{
\vliin{\cut}{}{\Theta_0, \ldots, \Theta_{n-1}, \Sigma}
{
\vliiiiq{\cut}{}{\Theta_0, \ldots, \Theta_{n-1}, \phi(\nu X\phi)}{
\vldr{d_0}{\Delta, \phi(\nu X\phi)}
}
{
\vldr{e_0}{\theta^\perp_0, \Theta_0}
}
{\vlhy{\ldots}}
{
\vldr{e_{n-1}}{\theta^\perp_{n-1}, \Theta_{n-1}}
}
}
{\vldr{e}{\Sigma, \phi^\perp(\mu X \phi^\perp)}}
}
$$
by~\Cref{rem:expansion} and inspecting the definition of $\GDer{}$ we conclude that the above derivation is not in $\Rset$. Therefore, 
 $d_0$ is not reducible. 

The second claim is argued as follows. Let $\nu Y_0 \psi_0 > \dotsm > \nu Y_p \psi_p$ be the $\nu$-subformulas of \( \nu X \phi \) ordered by the priority ordering (see \cref{defn:priority}). Let \( j \) be such that $\nu X \phi = \nu X_j \psi_j$ and $\vec \alpha_n = \vec{\gamma}\alpha \vec{\delta}$ where \( \lvert \vec \gamma \rvert = j \). Suppose $d$ is not $(\vec \alpha_0, \ldots,\vec \alpha_n)$-reducible. By \cref{reduc-approx-charac} there exists derivations \( e_j \in \inter{\theta_j^{\vec\alpha_j}}{\perp} \) such that
$$\cutform{d}{e_0, \ldots, e_{n-1}}{}\not \in \inter{(\nu X\phi)^{\vec \alpha_n}}{}.
$$
As \( \inter{(\nu X\phi)^{\vec \alpha_n}}{} = \inter{\nu^\alpha X\, \phi^{\vec \gamma}}{} \), there exists \( \beta < \alpha \) and \( e \in \inter{\phi(\nu X\phi)^{\vec \gamma \beta}}{} \) such that the derivation presented above in \eqref{eqn:nu-reduc-approx} is not in \( \Rset{} \). The same sequence of cut elimination steps shows that 
\( \cutform{d}{e_0, \dotsc, e_{n-1}, e}{} \not\in \Rset{} \), meaning that \( d_0 \) is not \( (\vec \beta_0, \dotsc, \vec \beta_{n}) \)-reducible where \( \vec\beta_i = \vec \alpha_i \) for \( i < n \) and \( \vec\beta_n = \vec \gamma \beta \).
     \end{itemize}
 \end{proof}

\begin{thm}[Soundness]\label{lem:progressing-implies-reducible}
     Let $\Rset\in\{\GDer{}, \Cnor{}\}$, and let $d$ be a progressing derivation with conclusion $\Gamma$. Then $d\in \inter{\Gamma}{}$.
\end{thm}
\begin{proof}
By~\Cref{prop:reducible-implies-in-interpretation}, it suffices to show that $d$ is  reducible. So, suppose towards contradiction that this is not the case. Fix ordinals assignments such that $d$ is not $(\vec \alpha_0, \ldots, \vec \alpha_n)$-reducible. Repeatedly applying \cref{lem:reflecting-non-totality} induces an infinite branch $\branch b$ of $d$ and a family of ordinal assignments $(\vec \alpha^i_0, \ldots, \vec \alpha^i_{n_i})_{i \in \omega}$ such that $d_{\branch{b}(i)}$ is not $(\vec \alpha^i_{0}, \ldots, \vec \alpha^i_{n_i})$-reducible. Without loss of generality we may assume that $(\vec \alpha^i_{0}, \ldots, \vec \alpha^i_{n_i})$ is the least such ordinal sequence for each \( i \) in the sense of \cref{prop:reducibility-approximations}. Since $d$ is progressing, $\branch b$ bears a good thread $\tau$ with, say, trace $t=(u_i, h_i, \phi_i)_{i}$. As the formulas enumerated by a thread are finite in number, there is some \( k\) such that \(\phi_i \fleq \phi_j\) for all \(i,j > k\) and, in particular, there is a finite bound the length of all sequences \( ( \vec \alpha_{h_i}^i )_i \). By \cref{lem:reflecting-non-totality}, $ \vec \alpha^{i+1}_{h_{i+1}}\leq \vec \alpha^{i}_{h_i} $ for all \( i \) and $ \vec \alpha^{i+1}_{h_{i+1}} < \vec \alpha^{i}_{h_i} $ infinitely often, contradicting the well-foundedness of ordinals.
\end{proof}

Using \Cref{prop:candidates-subset-R}, \Cref{lem:progressing-implies-reducible} implies immediately  that progressing proofs are also externally progressing and $\omega$-normalising. This allows us to establish our main result:

\begin{cor}[Cut elimination]\label{thm:cut-elimination}
    If $d \in \PDer{}$ then $d\in \Cnor{}$.
\end{cor}

As with the standard reducibility candidates methods for inductive systems, the cut elimination argument we above does not exhibit a concrete infinitary procedure for rewriting a derivation into a cut-free one. Also, it does not suggest any  approach for proving directly preservation of progressivity, which  is managed  implicitly by the notion of  $\Cnor{}$-reducibility.  In the next section we will develop  a second, more insightful cut elimination argument for~\Cref{thm:cut-elimination},  based on    $\GDer{}$-reducibility.

\section{Second cut elimination argument}
\label{sec:second-cut-elimination}
  \Cref{lem:progressing-implies-reducible} shows that progressing derivations are also externally progressing. Our second cut elimination argument for~\Cref{thm:cut-elimination} will rely on this result, where the external progressivity condition allows  a   straightforward proof of both  the existence of the limit cut-free proof (\textit{productivity}) and the sought-after preservation of  progressivity. We will formulate this result within a standard cut elimination procedure, where cuts  are stepwise pushed upward by the cut elimination rules, giving priority to the bottommost ones. To avoid technicalities related to cut commutations, the procedure will actually manipulate so-called \emph{multicut rules}, a generalisation of the cut rule representing series of consecutive cuts  (see, e.g., \cite{Baelde2016InfinitaryPT}).

\subsection{Multicut $\omega$-reduction sequences and paths}

As mentioned earlier, a multicut can be seen as a macro rule for a tree of consecutive cuts. The following definition formalises this idea.

\begin{defn}[Multicuts] 
 Let $d$ be a derivation with conclusion $\Gamma$, and  $u_1, \ldots, u_n$ with $n>1$ be nodes of $d$. 
 A \emph{multicut in \( d \)} is a tuple $\langle v; u_1, \ldots, u_n\rangle$ where $d_v$ is obtained from $d_{u_1}, \ldots, d_{u_n}$ by applying a series of consecutive cut rules. When $n=2$  the multicut is called  \emph{cut}. Multicuts  will be denoted $\mathsf{m}$. We call $u_1, \ldots, u_n$ (resp., $v$) the \emph{premises} (resp., \emph{conclusion}) of the multicut. 
A \emph{cut pair} of a multicut is a pair $(\phi, \phi^\perp)$ of dual formulas  occurring among the premises of the multicut such that, by permuting some cuts, a new multicut can be obtained that includes a cut rule $\rrule$ with cut formulas  $\phi$ and $\phi^\perp$. In this case, we call $\phi$ and $\phi^\perp$ \emph{cut formulas of the multicut}.
 
 The \emph{tree of $\mathsf{m} = \langle v; u_1, \ldots, u_n\rangle$ (in $d$)} is the set $\treem{\mathsf{m}}{d}=\{w \in d \mid \exists i \in \{1, \ldots, n\}(v \leq w \leq u_i)\}$, and the \emph{initial segment of $\mathsf{m}$ (in $d$)} is the set  $\initm{\mathsf{m}}{d}=\{w \in d\mid w < v  \}$.  
\end{defn}

\begin{exmp}
    Consider the following derivation $d$:
\[
\vlderivation{
\vlin{}{\parr}{\Gamma, \Sigma, \Gamma', \Sigma', \Theta, \alpha\parr \beta}
{
\vliin{}{\cut}{\Gamma, \Sigma, \Gamma', \Sigma', \Theta, \alpha, \beta}
{
\vliin{}{\cut}{\Gamma, \Sigma, \Gamma', \Sigma', \chi}
    {\vliin{}{\cut}{\Gamma, \Sigma, \psi}
    {\vldr{d_1}{\Gamma, \phi, \psi}}
    {\vldr{d_2}{\Delta, \phi^\perp}}}
    {\vliin{}{\cut}{\Gamma', \Sigma', \psi^\perp, \chi}
    {\vldr{d_3}{\Gamma', \psi^\perp, \delta}}
    {\vldr{d_4}{\Sigma', \chi, \delta^\perp}}}
}
{\vldr{d_5}{\Theta, \chi^\perp,\alpha, \beta}}
}
}
\]  
Then, we have $\mathsf{m}=\langle \langle 0 \rangle;   \langle 0,0,0,0\rangle, \langle 0, 0,0,1\rangle, \langle 0,0,1,0\rangle \langle 0,0,1,1\rangle, \langle 0,1\rangle\rangle$ and 
 $\mathsf{m}'=\langle \langle 0 \rangle;   \langle 0,0\rangle, \langle 0,1\rangle\rangle$ 
are  multicuts. Their trees are the following sets 
\[
\begin{array}{rcl}
  \treem{\mathsf{m}}{d}   & = &\begin{Bmatrix}
    \langle 0\rangle, \langle 0,0 \rangle,\langle 0,1 \rangle ,\langle 0,0,0 \rangle,\langle 0,0,1 \rangle ,
    \\ 
    \langle 0,0,0,0\rangle, \langle 0,0,0,1\rangle, \langle 0,0,1,0\rangle \langle 0,0,1,1\rangle
\end{Bmatrix} \\ \\ 
  \treem{\mathsf{m}'}{d}   & = &\{\langle 0 \rangle,   \langle 0,0\rangle, \langle 0,1\rangle \}
\end{array}
\]
and so $\treem{\mathsf{m}'}{d}\subseteq \treem{\mathsf{m}}{d}$. 
Their initial segments in $d$ are $\initm{\mathsf{m}}{d}=\initm{\mathsf{m}'}{d}=\{\langle  \rangle\}$. The cut pairs of $\mathsf{m}$ are $(\phi, \phi^\perp)$, $(\delta, \delta^\perp)$, $(\psi, \psi^\perp)$, and $(\chi, \chi^\perp)$, while $\mathsf{m}'$  only has $(\chi, \chi^\perp)$ as cut pair.
\end{exmp}

We can now introduce a multicut-based counterpart of $\omega$-reduction sequences, called \emph{multicut $\omega$-reduction sequences}. In particular,  we will focus on \emph{multicut $\omega$-reduction paths}, which track the  evolution of a specific multicut along those sequences. The latter notion requires the following preliminary definition.

\begin{defn}[Expansion vs reduction]
Let $d$ be a derivation, and let $\mathsf{m}$ be a multicut of $d$. We say that  $\mathsf{m}$ is \emph{expanded} to $\mathsf{m}'$ (and that $\mathsf{m}'$ is the \emph{expansion of $\mathsf{m}$}), written $\mathsf{m} < \mathsf{m}'$, if $ \mathsf{m}'$ is a multicut of $d$
 such that $\treem{\mathsf{m}}{d}\subsetneq \treem{\mathsf{m}'}{d}$. We write $\mathsf{m}\leq \mathsf{m}'$ if $\mathsf{m}< \mathsf{m}'$ or $ \mathsf{m}= \mathsf{m}'$.

 We also say that $\mathsf{m}$   \emph{reduces to $\mathsf{m}'$ (along $d \cutelims e$)}, written  $\mathsf{m}\mapsto \mathsf{m}'$, if $d \cutelims e$ rewrites $\mathsf{m}$ to a multicut $\mathsf{m}'$ of $e$, and $\mathsf{m}'$ is not a mere permutation of the cuts of $\mathsf{m}$. In this case we call $\mathsf{m}$ a \emph{redex}, and  $\mathsf{m}'$ a \emph{reduct of $\mathsf{m}$}.

\end{defn}

Note that there might be many (possibly zero) reducts of a multicut.

\begin{defn}[Multicut $\omega$-reduction sequence]
A \emph{multicut $\omega$-reduction sequence (from $d$)} is an $\omega$-indexed sequence of pairs $\mathfrak m=(d_i, \mathcal{M}_i)_{i \in \omega}$  where:
\begin{itemize}
\item $\reds s:= (d_i)_{i \in \omega}$ is a  $\omega$-reduction sequence
\item each $\mathcal{M}_i$ is a   (finite) set of multicuts of $d_i$ such that:
\begin{itemize}
    \item $\mathcal{M}_0$ is set of all cuts $\langle v; u_1, u_2 \rangle$ of $d_0$ such that $v$ has smallest length
    \item  $\mathcal{M}_{i+1}$ is obtained from $\mathcal{M}_i$ by replacing a multicut with either its expansion, if $d_i=d_{i+1}$, or its reducts along $d_i \cutelims d_{i+1}$ otherwise.
    \item $\mathcal{M}_i= \emptyset$ implies $d_i$ cut-free.
\end{itemize}
\item $\mathfrak s$ only applies cut elimination rules to multicuts in $(\mathcal{M}_i)_{i \in \omega}$.
 \end{itemize}
 We say that $\mathfrak m$ is \emph{terminating} if $\mathcal{M}_i= \emptyset$ for some $i \in \omega$. 
\end{defn}

\purple{Intuitively, terminating  multicut $\omega$-reduction sequences $\mathfrak m$  correspond to  \emph{finite} cut elimination processes. Notice that,  if its $\omega$-reduction sequence if non-strict, then $\mathfrak m$ might represent a cut elimination procedure that keeps ``expanding'' multicuts, i.e.,  merging more and more cuts into a single multicut. An example is $\mathfrak m=(d_i, \mathcal{M}_i)_{i \in \omega}$ such that, for all $i \in \omega$,  $d_i=d_{i+1}$ and  $\mathcal{M}_i=\{\mathsf{m}_i\}$  with  $\mathsf{m}_i < \mathsf{m}_{i+1}$. }

\begin{defn}[Multicut $\omega$-reduction path]
    A  \emph{multicut $\omega$-reduction path} is a family $\mathfrak p=(d_i, \mathsf{m}_i)_{i\in \omega}$ where, for all $i \in \omega$:
    \begin{itemize}
        \item $\reds s:= (d_i)_{i \in \omega}$ is a  $\omega$-reduction sequence
    \item $\mathsf{m}_i$ is a multicut in $d_i$ and   either $\mathsf{m}_i\leq  \mathsf{m}_{i+1}$, or  $\mathsf{m}_i\mapsto \mathsf{m}_{i+1}$.
\end{itemize}
We say that $\mathfrak p$ is  \emph{proper} if $\mathsf{m}_i\neq \mathsf{m}_{i+1}$ for all infinitely many $i \in \omega$. Finally,  given a multicut $\omega$-reduction sequence $\mathfrak m=(d_i, \mathcal{M}_i)_{i\in \omega}$, we write $\mathfrak p \in \mathfrak m $ if $\mathsf{m}_i\in \mathcal{M}_i$ for all $i \in \omega$.
\end{defn}

\begin{rem}\label{rem:path-from-sequences}
    By K\"{o}nig's lemma, for every non-terminating multicut $\omega$-reduction sequence  $(d_i, \mathcal{M}_i)_{i\in \omega}$ there is a proper multicut $\omega$-reduction path $\mathfrak p \in \mathfrak m$.
\end{rem}

We will work with so-called \emph{fair}  multicut $\omega$-reduction sequences, where every multicut is eventually either expanded or reduced.

\begin{defn}[Fairness]
    We say that a multicut $\omega$-reduction sequence $\mathfrak m=(d_i, \mathcal{M}_i)_{i \in \omega}$ is \emph{fair} if every $\mathfrak p \in \mathfrak m$ is proper.
\end{defn}

\subsection{IC sets from multicut $\omega$-reduction paths}

A proper multicut reduction path from  a  (possibly non-progressing) derivation $d$ can  be seen as a pointer machine that visits a subtree of $d$, called \emph{covering} of $d$, written $\mathcal{C}_{\mathfrak p}$. In this subsection we show that the collection of branches of this subtree forms an internally closed set.

\begin{defn}[Frontier and covering]\label{defn:covering}
Let $\mathfrak p = (d_i, \mathsf{m}_i)_{i\in \omega}$ be a  proper multicut $\omega$-reduction path from $d$. The \emph{$i$-frontier} of $\mathfrak p$, written $\mathcal{F}_{\mathfrak p}^i$, is a tuple of nodes of $d_0$ defined by induction on $i$ as follows:
\begin{itemize}
    \item $\mathcal{F}_{\mathfrak p}^0\dfn \langle u_1, \ldots, u_n \rangle$, where $ u_1, \ldots, u_n$ are the premises of $\mathsf{m}_0$.
    \item Suppose $\mathcal{F}_{\mathfrak p}^i=\langle w_1, \ldots w_n \rangle$. We have two cases:
    \begin{enumerate}
        \item  $\mathsf{m}_{i}\mapsto \mathfrak {m}_{i+1}$, with  $\mathsf{m}_i=\langle v ; u_1, \ldots, u_n \rangle$. Without loss of generality, we can reduce $d_i \cutelims d_{i+1}$ to the following cases:
        \begin{itemize}
            \item $d_i\cutelim d_{i+1}$ is one step of the cut elimination rule commuting two cuts, which permutes the $j$-th and the $j+1$-th premises of $\mathsf{m}_i$. Set  $$\mathcal{F}_{\mathfrak p}^{i+1}\dfn \langle w_1, \ldots, w_{j-1}, w_{j+1}, w_j, w_{j+2}, \ldots w_n \rangle. $$
            \item $d_i\cutelim d_{i+1}$ applies a critical cut elimination step to a cut in $\mathsf{m}_i$ with  premises $u_j$ and $u_{j+1}$. Let $\mathsf{m}_{i+1}=\langle v' ; u'_1, \ldots, u'_m\rangle$. Set $\mathcal{F}_{\mathfrak p}^{i+1}\dfn \langle w'_1, \ldots, w'_{m}\rangle$ where:
            $$
                 w'_t \dfn \begin{cases}
                     w_t, &\text{if }u'_t=u_t,\\
                     w_t0, &\text{if }u'_t=u_t0,\\
                      w_t1, &\text{if }u'_t=u_t1.
                 \end{cases}
           $$
           \item $d_i \cutelim d_{i+1}$ applies a non-critical  reduction to the inference rule with conclusion $u_j$ so that 
           $\initm{d_{i}}{\mathsf{m}_{i}}=\{\langle \rangle, \langle b_1\rangle, \ldots, \langle b_1, \ldots, b_k \rangle\}$ and $\initm{d_{i+1}}{\mathsf{m}_{i+1}}=\{\langle \rangle, \langle b_1\rangle, \ldots, \langle b_1, \ldots, b_{k+1} \rangle\}$.  Set $\mathcal{F}_{\mathfrak p}^{i+1}\dfn \langle w'_1, \ldots, w'_{n}\rangle$ where:
           $$
                 w'_t \dfn \begin{cases}
                     w_t, &\text{if }t \neq j,\\
                     w_t0, &\text{if }t=j \text{ and }b_{k+1}=0,\\
                      w_t1, &\text{if }t=j \text{ and }b_{k+1}=1.
                 \end{cases}
$$
        \end{itemize}
        \item $\mathsf{m}_i \leq \mathsf{m}_{i+1}$. W.l.o.g. we can assume that $\mathsf{m}_i < \mathsf{m}_{i+1}$ and $\mathsf{m}_{i+1}$ is as follows:
         $$
 \small
 \vlderivation{
\vliiiq{}{\mathsf{m}_{i+1}}{\Gamma}{\vldr{d_1}{\Gamma_1}}{\vlhy{\cdots}}{ \vldr{d_{n}}{\Gamma_{n}}}
} 
= 
\vlderivation{
\vliiiiiq{}{\mathsf{m}_i}{\Gamma}{\vldr{d_1}{\Gamma_1}}{\vlhy{\cdots}}{\vliin{\cut}{\rrule}{\Gamma'_j}{\vldr{d_j}{\Gamma_j}}{\vldr{d_{j+1}}{\Gamma_{j+1}}}}{\vlhy{\cdots}}{ \vldr{d_{n}}{\Gamma_{n}}}
} 
$$
We set $\mathcal{F}_{\mathfrak p}^{i+1}\dfn \langle w_1, \ldots, w_{j-1}w_j0,w_j1, w_{j+1}, \ldots, w_n\rangle$.
 \end{enumerate}
\end{itemize}
The  \emph{covering} of $\mathfrak p$, written $\mathcal{C}_{\mathfrak p}$, is the subtree of $d$ obtained by the prefix-closure of the set of nodes appearing in $(\mathcal{F}_{\mathfrak p}^{i})_{i \in \omega}$.
\end{defn}

\begin{lem}\label{lem:traversed-implies-ICset} 
Let $d$ be a (possibly non-progressing) derivation, and let $\mathfrak p$ be a proper multicut $\omega$-reduction path from $d$. Then, $\Branch{\mathcal{C}_{\mathfrak p}}$  is internally closed.
\end{lem}
\begin{proof}
Let $\mathfrak p=(d_i, \mathsf{m}_i)_{i \in \omega}$ be a proper multicut $\omega$-reduction path from $d$. By definition of proper multicut $\omega$-reduction path,  $\mathcal{C}_{\mathfrak p}$ is infinite, and so it has at least one branch by  weak K\"{o}nig's lemma. We now show that $\Branch{\mathcal{C}_{\mathfrak p}}$ is an IC set. Closure under infima and suprema is straightforward, so we only need to prove the coherence condition for IC sets. Let $\branch b \in \Branch{\mathcal{C}_{\mathfrak p}}$ and $\tau$ internal thread of  ${\branch b}$, and let $t=(p_i,\branch{b}(n+i), \phi_i)_{i\in \omega}$ be the trace of $\tau$, such that  $(p_0,\branch{b}(n), \phi_0)$ is a cut formula. By definition, every $ (p_i,\branch{b}(n+i), \phi_i)$ is a cut formula of a multicut in $\mathfrak p$. Let $s=(s(i))_{i \in \lambda}$ with $\lambda \leq \omega$ be the sequence of formulas that appear in a cut pair with  formulas of $t$ in the multicuts of $\mathfrak p$. Notice that $s(i+1)\flleq s(i)$ for all $i \in \lambda$. By definition, since $\tau$ is a thread, and every formula occurrence in $t$ is a cut formula of a multicut in $\mathfrak p$, there are infinitely many $(i, j)$ such that $s(i)$ and $t(j)$ appear in the same  cut pair of a multicut of $\mathfrak p$, and $t(j)$ is a principal formula.
This means that, for infinitely many $i$, there is  $d_i \cutelim^*d_{i+1}$ where a  critical cut elimination step is applied to the cut pair $(s(i), t(j))$, and so $s(i)$ is a principal formula for infinitely many $i$. Therefore, it must be that $ s(i+1)\fll s(i)$ for infinitely many $i \in \lambda$, and so $\lambda=\omega$. This means that $s$ is the expansion of a thread $\sigma$, which implies $\sigma=\tau^\perp$. Moreover,  $\sigma$ must be bourne by a branch  $\branch c$ of $d$. By definition, $\branch c\in \Branch{\mathcal{C}_{\mathfrak p}}$ and   $\ortcoh{\tau}{\branch b }{\branch c }$.
\end{proof}

	Notice that different multicut $\omega$-reduction sequences will induce distinct  collections of  proper multicut $\omega$-reduction paths, which in turn might traverse distinct  IC sets.

\subsection{(Externally) progressing derivations are $\omega$-normalisable}

In this subsection we present the second cut elimination argument of this paper. First, we show how to construct a fair   multicut $\omega$-reduction sequence  $\mathfrak m$ from a given derivation $d$. Second, we show that, if  $d$ is externally progressing,   $\mathfrak m$ defines a depth-increasing  $\omega$-reduction sequence to a cut-free progressing derivation. We conclude by appealing  to \Cref{lem:progressing-implies-reducible}, which implies that progressing derivations are also externally progressing.

\begin{restatable}{lem}{lemmult}\label{lem:construction-of-multicut-reduction}
    Let $d$ be a  derivation with conclusion $\phi_1, \ldots, \phi_k$ (possibly non progressing). Then, there exists a fair  multicut $\omega$-reduction sequence from $\cutform{d}{\id_{\phi_1}, \ldots,  \id_{\phi_k}}{}$.
\end{restatable}
\begin{proof}
We construct by induction on $i \in \omega$ a \emph{fair} multicut $\omega$-reduction sequence $\mathfrak m=(d_i, \mathcal{M}_i)_{i\in \omega}$ that satisfies the following condition for all $i \in \omega$:
\begin{itemize}\label{eqn:condition}  
    \item[($\star$)]  Every multicut in $\mathcal{M}_i$ has an identity derivation among its premises.
\end{itemize}

For the base case, we set $d_0 \dfn \cutform{d}{\id_{\phi_1}, \ldots,  \id_{\phi_k}}{}$ and let $\mathcal{M}_0$ be the set containing a single multicut formed by collecting the bottommost \( k \) cuts in $d_0$. Condition ($\star$) is clearly satisfied. 

Concerning the  inductive step,  if $\mathcal{M}_{i}= \emptyset$ then we simply set $d_{i+1}\dfn d_{i}$ and $\mathcal{M}_{i+1}\dfn \emptyset$. Otherwise, we consider a multicut $\mathsf{m}=\langle v; u_1, \ldots, u_n\rangle\in \mathcal{M}_i$ of the following shape\footnote{We will omit details relative to fairness in the construction of the   multicut $\omega$-reduction sequence. However,  our definition can be easily adapted to satisfy fairness.}
$$
 \small
\vlderivation{
\vliiiq{\mcut}{\mathsf{m}}{\Gamma}{\vldr{e_1}{\Gamma_1}}{\vlhy{\cdots}}{ \vldr{e_{n}}{\Gamma_{n}}}
}
$$
We perform a case analysis on $\mathsf{m}$, checking that  condition ($\star$) is  preserved. 
\begin{enumerate}[(I)]
    \item If there is $j \in \{1, \ldots, n\}$  such that $u_j$ is the conclusion of a cut rule, then we set $d_{i+1}:=d_i$ and  $\mathcal{M}_{i+1}\dfn (\mathcal{M}_i \setminus\{\mathsf{m}\})\cup \{\mathsf{m}'\}$, where $\mathsf{m}'\dfn \langle v; u_1, \ldots, u_{j-1}, u_j0, u_j1, u_{j+1}, \ldots, u_n\rangle$. 
    Notice that $\mathsf{m}<\mathsf{m}'$.

 \item Otherwise, for all $j \in \{1, \ldots, n\}$, $u_j$ is the conclusion of a rule $\rrule_j\neq \cut$. If there is $\rrule_j$ whose principal formula is not a cut formula of $\mathsf{m}$ then we have two subcases:
\begin{itemize}
    \item If $\rrule_j$ is unary then $\mathsf{m}$ has the following shape
   $$
     \vlderivation{
\vliiiiiq{\mcut}{\mathsf{m}}{\Gamma}{\vldr{e_1}{\Gamma_1}}{\vlhy{\cdots}}{\vlin{\rrule_j}{}{\Gamma_j}{\vldr{e'_j}{\Gamma'_j}}}{\vlhy{\cdots}}{ \vldr{e_n}{\Gamma_{n}}}
} 
$$
 by applying  non-critical cut elimination steps we obtain:
$$
\vlderivation{
\vlin{\rrule_j}{}{\Gamma}{\vliiiiiq{\mcut}{\mathsf{m}'}{\Gamma'}{\vldr{e_1}{\Gamma^i_1}}{\vlhy{\cdots}}{\vldr{e'_j}{\Gamma'_j}}{\vlhy{\cdots}}{ \vldr{e_{n}}{\Gamma_{n}}}}
}
$$
 We set $d_{i+1}$ as the derivation above, where $\mathsf{m}\mapsto\mathsf{m}'$,  and $\mathcal{M}_{i+1}\dfn  (\mathcal{M}_i \setminus\{\mathsf{m}\})\cup \{\mathsf{m}'\}$. 
\item If $\rrule_j\in \{\otimes, \with\}$ and $\mathsf{m}$ has the following shape 
 $$
 \small
\vlderivation{
\vliiiiiq{\mcut}{\mathsf{m}}{\Gamma}{\vldr{e_1}{\Gamma_1}}{\vlhy{\cdots}}{\vliin{\rrule_j}{}{\Gamma_j}{\vldr{e'_j}{\Gamma'_j}}{\vldr{e''_j}{\Gamma''_j}}}{\vlhy{\cdots}}{ \vldr{e_{n}}{\Gamma_{n}}}
}
$$
 by applying  non-critical cut elimination steps we obtain:
$$
\small
 \vlderivation{
\vliin{\rrule_j}{}{\Gamma}{\vliiiiiq{\mcut}{\mathsf{m}'}{\Gamma'}{\vldr{e_1}{\Gamma^i_1}}{\vlhy{\cdots}}{\vldr{e'}{\Gamma'_j}}{\vlhy{\cdots}}{ \vldr{e_{n}}{\Gamma_{n}}}}{\vliiiiiq{\mcut}{\mathsf{m}''}{\Gamma''}{\vldr{e_1}{\Gamma_1}}{\vlhy{\cdots}}{\vldr{e''}{\Gamma''_j}}{\vlhy{\cdots}}{ \vldr{e_{n}}{\Gamma_{n}}}}
}
    $$
We set $d_{i+1}$ as the derivation above, where $\mathsf{m}\mapsto\mathsf{m}'$ and $\mathsf{m}\mapsto\mathsf{m}''$, and $\mathcal{M}_{i+1}\dfn (\mathcal{M}_i \setminus\{\mathsf{m}\})\cup \{\mathsf{m}', \mathsf{m}''\}$.  
\end{itemize}
\item Otherwise, for all $j \in \{1, \ldots, n\}$, $u_j$ is the conclusion of a rule $\rrule_j\neq \cut$ whose principal formula is a cut formula of $\mathsf{m}$. Then,  by applying a series of non-critical cut elimination steps permuting cuts, we obtain  the following multicut $\mathsf{m}'$:
$$
 \small
\vlderivation{
\vliiiq{\mcut}{\mathsf{m}'}{\Gamma}{\vldr{e'_1}{\Gamma'_1}}{\vlhy{\cdots}}{ \vldr{e'_{n}}{\Gamma'_{n}}}
}
$$
where, for some $1 \leq t\leq n-1$, $\mathsf{m}'$ contains the  following critical cut:
$$
\vlderivation{
\vliin{\cut}{\rrule}{\Delta}{\vldr{e'_{t}}{\Gamma'_{t}}}{\vldr{e'_{t+1}}{\Gamma'_{t+1}}}
}
$$
We define $d_{i+1}$ as the derivation obtained by applying the corresponding critical cut elimination step to $\rrule$. As for the definition of $\mathcal{M}_{i+1}$, we have two subcases:
\begin{itemize}
 \item $\mathsf{m}'\mapsto\mathsf{m}''$ along the critical cut elimination step for $\rrule$. Notice that this reduct must be \emph{unique}. Then, we set
 $\mathcal{M}_{i+1}\dfn(\mathcal{M}_i \setminus\{\mathsf{m}\})\cup \{\mathsf{m}''\}$. Notice that  $\mathsf{m}\mapsto\mathsf{m}''$.
    \item  $\mathsf{m}'$ has no reduct along the critical cut elimination step. By \Cref{rem:at-least-one-reduct}, it must be that  $n=2$, i.e., $\mathsf{m}=\mathsf{m}'$ is the following critical cut 
 $$
\vlderivation{
\vliin{\cut}{\rrule}{\Gamma_2}{\vlin{\unit}{}{\unit}{\vlhy{}}}{\vlin{\perp}{}{\perp, \Gamma_2}{\vldr{e_2}{\Gamma_2}}}
}
  $$
By the inductive hypothesis,  condition ($\star$) ensures that  $\Gamma_2= \unit$ and  $e_2$ is just the rule $\unit$. Therefore, we set $\mathcal{M}_{i+1}\dfn\mathcal{M}_i \setminus\{\mathsf{m}\}$. Notice that, if $\mathcal{M}_{i+1}= \emptyset$ then $d_{i+1}$ is just the rule $\unit$, and so it is cut-free.
\end{itemize}
\end{enumerate}
\end{proof}

\begin{restatable}{lem}{lemtermult}\label{lem:trivial-implies-normalisable}
    Let  $\mathfrak m$ be a terminating multicut $\omega$-reduction sequence from $d$. If $d \in \GDer{}$ then $d$ is $\omega$-normalisable. 
\end{restatable}
\begin{proof}
    Let $\mathfrak m=(d_i, \mathcal{M}_i)_{i \in \omega}$ with $d_0=d$. Since $\mathfrak m$ is terminating, then $\mathcal{M}_i= \emptyset$ for some $i \in \omega$ and so, by definition,  $d_i$ is cut-free. This means that $\mathfrak s=(d_i)_{i \in \omega}$ is non-strict and trivially depth-increasing with  cut-free limit $d_i$. Also, the limit must be externally progressing by~\Cref{rem:ext-prog-local-preserv}. Finally, since $d_i$ is cut-free then it is also progressing by~\Cref{prop:cut-free-good-equals-progressing}.
\end{proof}

\begin{thm}
\label{lem:good-norm}
    If $d \in \GDer{}$ then $d \in \Cnor{}$.
    \end{thm}
    \begin{proof}
We show that if $\cutform{d}{\id_{\phi_1}, \ldots,  \id_{\phi_k}}{}$ is externally progressing then it is also  $\omega$-normalisable. This will allow us to conclude, since  by~\Cref{prop:id-in-tot-candidate}.\ref{enum:id-cond-1} we have $
d \in \GDer{} \Leftrightarrow \cutform{d}{\id_{\phi_1}, \ldots,  \id_{\phi_k}}{} \in \GDer{} \Rightarrow \cutform{d}{\id_{\phi_1}, \ldots,  \id_{\phi_k}}{} \in \Cnor{} \Leftrightarrow d \in \Cnor{}
$.

By~\Cref{lem:construction-of-multicut-reduction}, there exists a fair  multicut $\omega$-reduction sequence $\mathfrak m=(d_i, \mathcal{M}_{i})_{\omega}$ from $\cutform{d}{\id_{\phi_1}, \ldots,  \id_{\phi_k}}{}$. If $\mathfrak m$ is terminating then we conclude by~\Cref{lem:trivial-implies-normalisable}. Otherwise, we set  $\initm{\mathcal{M}_i}{d_i}\dfn  \bigcup_{\mathsf{m} \in \mathcal{M}_i}\initm{\mathsf{m}}{d_i}$. By definition of $\mathfrak m$:
\begin{itemize}
    \item each $\initm{\mathcal{M}_i}{d_i}$ is a tree labelled with sequents and rules except $\cut$
    \item  $(\initm{\mathsf{m}_i}{d_i})_{i \in \omega}$ defines a non-decreasing $\subseteq$-chain. 
\end{itemize}
Let $T\dfn \bigcup_{i \in \omega} \initm{\mathcal{M}_i}{d_i}$, and let  $B_{\mathfrak p}\dfn \bigcup_{i \in \omega} \initm{\mathsf{m}_i}{d_i}$ for every multicut $\omega$-reduction path $\mathfrak p \in \mathfrak m$.  By construction, $T$ is a tree and the cut rule does not label any of its nodes. We need to show that $T$ is a derivation and, moreover, that it is progressing. 

Now, by \Cref{rem:path-from-sequences}, since $\mathfrak m$ is a non-terminating  multicut $\omega$-reduction sequence the set of     proper multicut $\omega$-reduction paths $\mathfrak p \in \mathfrak m$ is non-empty. Moreover, by fairness, all such  $\mathfrak p$ are proper. So, let us consider $\mathfrak p\in \mathfrak m$. By~\Cref{lem:traversed-implies-ICset},   $\Branch{\mathcal{C}_{\mathfrak p}}$ is internally closed. By external progressivity of $d_0$, there is a branch $\branch b \in \Branch{\mathcal{C}_{\mathfrak p}}$ bearing a good external thread $\tau$. Let  $t=(p_j, \branch{b}(j), \theta_j)_{j \in \omega}$ be the trace of $\tau$. By definition of $\mathcal{C}_{\mathfrak p}$ there is  a monotone non-decreasing function $f:\omega \to \omega$ such that $\mathsf{m}_i$ has the formula $\theta_{f(i)}$ in one of its premises, and $f(i)<f(i+1)$ for infinitely many $i$.
 Since $\tau$ is good, $\theta_i$ is principal for a rule $\rrule_i$ for infinitely many $i$, and cannot be a cut formula of a multicut, since $\tau$ is external.  Therefore, there are infinitely many $\mathsf{m}_i\mapsto\mathsf{m}_{i+1}$ where each such rule $\rrule_i$ is permuted downward by the non-critical rules in such a way that $\initm{\mathsf{m}_i}{d_i}\subsetneq \initm{\mathsf{m}_{i+1}}{d_{i+1}}$, and so  $B_{\mathfrak p}$  is an infinite set of (labelled) nodes. Consider the  \emph{unique} branch ${\branch c}_{\mathfrak p}$ of $T$ given by ${\branch c}_{\mathfrak p}: \omega\to \omega^{\leq}$ such that ${\branch c}_{\mathfrak p}(0)=\langle \rangle$ and ${\branch c}_{\mathfrak p}(n+1)=\langle b_1, \ldots, b_n \rangle\in B_{\mathfrak p}$. By construction, we have that ${\branch c}_{\mathfrak p}$ is a branch of $T$ bearing the good external thread $\tau$. 

 Therefore, $T$ must be a derivation $d^*$, and all its branches bear a good (external) thread, that is, $d^*$ is progressing. This implies that $\mathfrak{s}=(d_i)_{i \in \omega}$ is depth-increasing and its limit is $d^*$. So, $d_0$ is $\omega$-normalising.
 % By fairness and the fact that any multicut is reducible (as one of the cases (I)-(III) always applies), we have that every multicut in $\mathcal{M}_i$ occurs in a multicut reduction path of $\mathsf{M}$. Furthermore, by minimality, for every  multicut reduction path $(d_i, \mathsf{m}_i)_{i \in \omega}\in \mathsf{M}$,  $\initm{\mathsf{m}_i}{d_i}$ is cut-free, and $\initm{\mathsf{m}_i}{d_i}\subseteq \initm{\mathsf{m}_{i+1}}{d_{i+1}}$. We conclude by applying \cref{lem:ICset-implies-continuously-normalisable}, noticing that each good thread in $\branch c$ must be \textit{external} because the multicut reduction sequence is minimal. 
\end{proof}

Using~\Cref{lem:progressing-implies-reducible} and \Cref{prop:candidates-subset-R}, \Cref{lem:good-norm} implies immediately  that progressing proofs are also externally progressing. This allows us to re-establish \Cref{thm:cut-elimination}.

\section{Conclusions and future work}

In this paper, we developed cut elimination methods for ill-founded $\muMALL$ based on Tait and Girard’s reducibility candidates. Our result  addresses one of the critical aspects of ill-founded proof theory, namely the interaction between infinitary cut elimination and global correctness criteria, and  provides a robust and modular framework for reasoning about infinitary proofs.

We  view this work as a step toward a uniform and proof-theoretic approach to cut elimination for ill-founded systems. To this end we envisage adapting our results to other contexts, such as  intuitionistic logic, and (possibly higher-order) $\mu$-arithmetics. We also plan to study preservation of other global correctness criteria within our framework, such as Sprenger and Dam's semantic notion of run or automata-based conditions~\cite{sprengerdam03:journal}, and the notion of bouncing thread from~\cite{bouncing}.

\clearpage
\bibliographystyle{alpha}
\bibliography{biblio}

\clearpage
\appendix

\end{document}